\documentclass[a4paper,12pt,reqno]{amsart}
\usepackage{amsmath,amsthm,amsfonts}
\usepackage{amssymb}
\usepackage{hyperref}
\usepackage[T1]{fontenc}
\usepackage{times}
\usepackage{paralist}
\usepackage{graphicx}

\setdefaultenum{i)}{(a)}{}{}
\usepackage[round]{natbib}
\usepackage[margin=1.2in]{geometry}

\newtheoremstyle{myplain}
{}
{}
{\itshape}
{}        
{\scshape}
{}
{.5em}
{}

\newtheoremstyle{myremark}
{}
{}
{}
{}        
{\scshape}
{}
{.5em}
{}

\numberwithin{equation}{section}
\usepackage{enumerate}														
\usepackage{color}

\newtheorem{theorem}{Theorem}[section]

\newtheorem{definition}[theorem]{Definition}

\newtheorem{lemma}[theorem]{Lemma}
\newtheorem{proposition}[theorem]{Proposition}
\newtheorem{remark}[theorem]{Remark}

\numberwithin{equation}{section}

\newcommand{\eps}{\varepsilon}
\newcommand{\sgn}{\operatorname{sgn}}
\newcommand{\nada}[1]{}

\DeclareMathOperator{\avtrco}{\textup{ATC}}
\DeclareMathOperator{\esr}{\textup{ESR}}

\long\def\symbolfootnote[#1]#2{\begingroup\def\thefootnote{\fnsymbol{footnote}}\footnote[#1]{#2}\endgroup}

\title{The Limits of Leverage}

\begin{document}
\symbolfootnote[0]{For helpful comments, we thank Jak\v sa Cvitani\'c, Gur Huberman, Johannes Muhle-Karbe, Walter Schachermayer, Ronnie Sircar, Matt Spiegel, Ren\'e Stulz, Peter Tankov and seminar participants at ETH Z\"urich, S\'eminaire Bachelier Paris, Vienna Graduate School of Finance, Quant Europe, Quant USA, Banff Workshop on Arbitrage and Portfolio Optimization, Institute of Finance at USI (Lugano).}
\symbolfootnote[0]{$\dag$ Boston University, Department of Mathematics and Statistics, 111 Cummington Mall, Boston, MA 02215, USA, and Dublin City University, School of Mathematical Sciences, Glasnevin, Dublin 9, Ireland, \texttt{paolo.guasoni@dcu.ie}. Partially supported by the ERC (278295), NSF (DMS-1412529), SFI (16/IA/4443,16/SPP/3347).}
\symbolfootnote[0]{$\ddag$ University of Limerick, Department of Mathematics \& Statistics, Castletroy, Co. Limerick, Ireland, \texttt{eberhard.mayerhofer@ul.ie}. Partially supported by the ERC (278295) and SFI (08/SRC/FMC1389).}

\maketitle
\vspace{-.5cm}
\centerline{\textsc{Paolo Guasoni$^{\dag}$}}
\medskip\centerline{\textit{Boston University and Dublin City University}}
\medskip\medskip

\centerline{\textsc{Eberhard Mayerhofer $^{\ddag}$}}
\medskip\centerline{\textit{University of Limerick}}
\medskip\medskip\medskip\medskip\medskip
\small

\small
When trading incurs proportional costs, leverage can scale an asset's return only up to a maximum multiple, which is sensitive to its volatility and liquidity. In a model with one safe and one risky asset, with constant investment opportunities and proportional costs, we find strategies that maximize long term returns given average volatility. As leverage increases, rising rebalancing costs imply declining Sharpe ratios. Beyond a critical level, even returns decline. Holding the Sharpe ratio constant, higher asset volatility leads to superior returns through lower costs.



\medskip\medskip\medskip
\normalsize

\section{Introduction}

If trading is costless, leverage can scale returns without limits. Using the words of \cite{sharpe2011investors}:
\begin{quotation}
``If an investor can borrow or lend as desired, any portfolio can be leveraged up or down. A combination with a proportion $k$ invested in a risky portfolio and $1-k$ in the riskless asset will have an expected excess return of $k$ [times the excess return of the risky portfolio] and a standard deviation equal to $k$ times the standard deviation of the risky portfolio. Importantly, the Sharpe Ratio of the combination will be the same as that of the risky portfolio.''
\end{quotation}

In theory, this insight implies that the efficient frontier is linear, that efficient portfolios are identified by their common maximum Sharpe ratio, and that any of them spans all the other ones. 
Also, if leverage can deliver any expected return, then risk-neutral portfolio choice is meaningless, as it leads to infinite leverage. 

In practice, hedge funds and high-frequency trading firms employ leverage to obtain high returns from small relative mispricing of assets. A famous example is Long Term Capital Management, which used leverage of up to 40 times to increase returns from convergence trades between on-the-run and off-the-run treasury bonds (for example, see \cite{edwards1999hedge}). 

This paper shows that trading costs undermine these classical properties of leverage and set sharp theoretical limits to its applications. We start by characterizing the set of portfolios that maximize long term expected returns for given average volatility, extending the familiar efficient frontier to a market with one safe and one risky asset, where both investment opportunities and relative bid-ask spreads are constant. 
Figure \ref{fig:front} plots this frontier: expectedly, trading costs decrease returns, with the exception of a fully safe investment (the axes origin) or a fully risky investment (the attachment point with unit coordinates), which lead to static portfolios without trading, and hence earn their frictionless returns.\footnote{As we focus on long term investments, we neglect the one-off costs of set up and liquidation, which are negligible over a long holding period.}

But trading costs do not merely reduce expected returns below their frictionless benchmarks. Unexpectedly, in the leverage regime (to the right of the full-investment point) rebalancing costs rise so quickly with volatility that returns cannot increase beyond a critical factor, the leverage \emph{multiplier}. This multiplier depends on the relative bid-ask spread $\varepsilon$, the expected excess return $\mu$ and volatility $\sigma$, and approximately equals 
\begin{equation}\label{max leverage multiplier}
0.3815\left(\frac{\mu}{\sigma^2}\right)^{1/2}\varepsilon^{-1/2}.
\end{equation}

Table \ref{tab:mult} shows that even a modest bid-ask spread of 0.10\% implies a multiplier of 23 for an asset with 10\% volatility and 5\% expected return (similar to a long-term bond), while the multiplier declines to 10 for an asset with equal Sharpe ratio, but volatility of 50\% (similar to an individual stock). Leverage opportunities are much more limited for more illiquid assets with a spread of 1\%: the multiplier declines from less than 8 for 10\% volatility to less than 4 for 50\% volatility. 
Importantly, these limits on leverage hold even allowing for continuous trading, infinite market depth (any quantity trades at the bid or ask price), and zero capital requirements. 

\begin{table}[!t]
\centering
\begin{tabular}{ r l l l } 
& \multicolumn{3}{c}{Bid-Ask Spread ($\varepsilon$)}\\
Volatility ($\sigma$) & 0.01\% & 0.10\% & 1.00\% \\
\cline{2-4}\\
10\% &71.85\;(71.22)& 23.15\;(22.58)& 7.72\;(7.12)\\ 
20\%  &50.88\;(50.36) &16.45\;(15.92) &5.56\;(5.04) \\
50\% &32.30\;(31.85) &10.54\;(10.07)& 3.66\;(3.18)\\
\\
\end{tabular}
\caption{\label{tab:mult}
Leverage multiplier (maximum factor by which a risky asset's return can be scaled) for different asset volatilities and bid-ask spreads, holding the Sharpe ratio at the constant level of 0.5. Multipliers are obtained from numerical solutions of \eqref{eq: TKA fbp}, while their approximations from \eqref{max leverage multiplier} are in brackets.}
\end{table}
\begin{figure}[!ht]
\centering
\includegraphics[width=\textwidth]{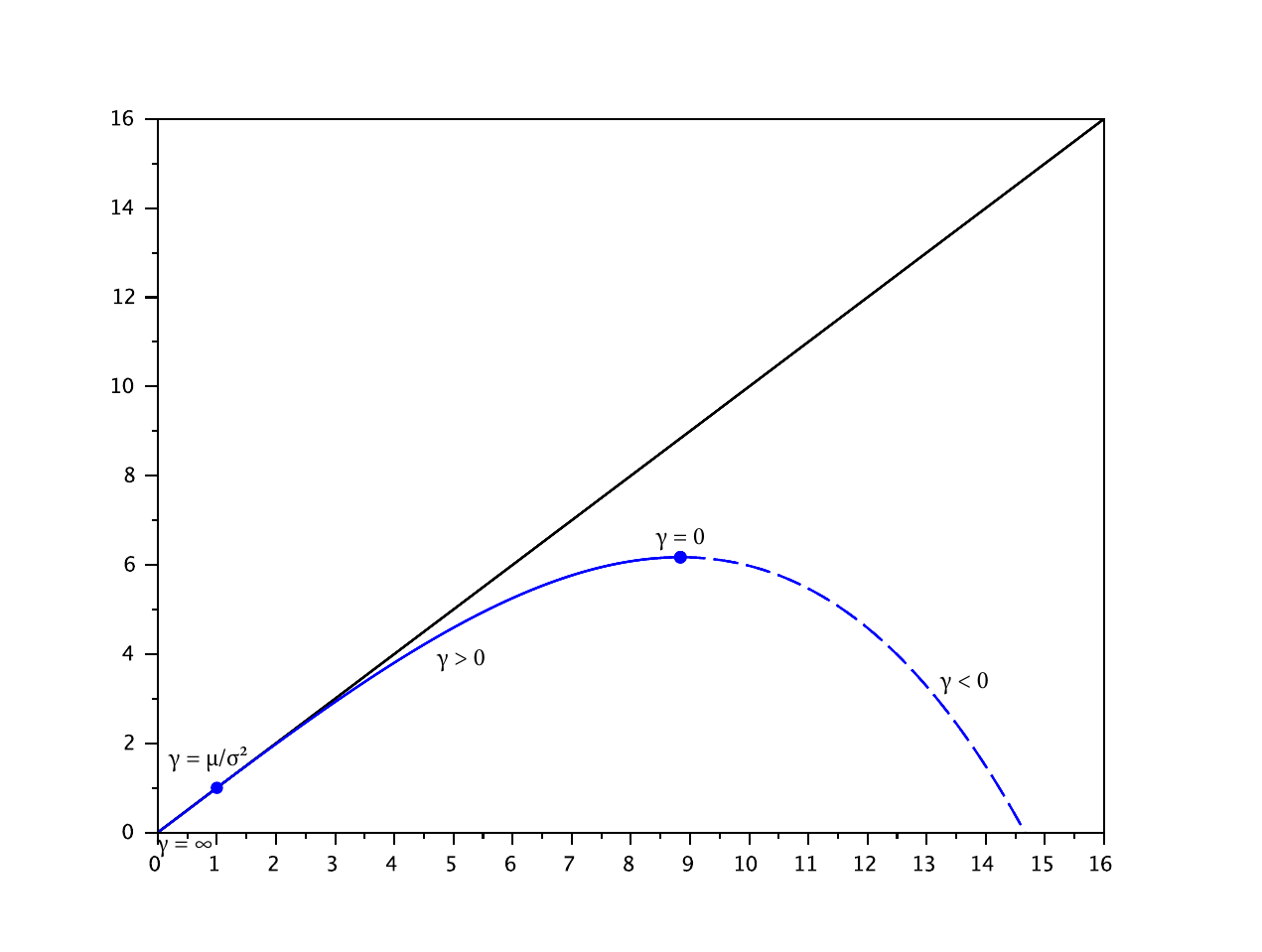}
\caption{\label{fig:front}
Efficient Frontier with trading costs, as expected excess return (vertical axis, in multiples of the asset's return) against standard deviation (horizontal axis, in multiples of the asset's volatility). The asset has expected excess return $\mu=8\%$, volatility $\sigma=16 \%$, and bid-ask spread of 1\%. Each point on the curve represents the performance of the optimal portfolio with risk aversion $\gamma$. The upper line denotes the classical efficient frontier, with no transaction costs. 
The maximum height of the curve ($\gamma=0$) corresponds to the leverage multiplier. As $\gamma$ increases, leverage, return, and volatility all decrease, reaching the asset's own performance $(1,1)$ at $\gamma=\mu/\sigma^2$. As $\gamma$ increases further, exposure to the asset declines below one, eventually vanishing at the origin ($\gamma=\infty$). The dashed frontier ($\gamma<0$) is not ``efficient'' in that such returns are maximal for given volatility, but can be achieved with lower volatility in the solid frontier ($\gamma>0$).
}
\end{figure}

Our results have two broad implications. First, with a positive bid-ask spread even a risk-neutral investor who seeks to maximize expected long-run returns takes finite leverage, and in fact a rather low leverage ratio in an illiquid market -- risk-neutral portfolio choice is meaningful. The resulting multiplier sets an endogenous level of risk that the investor chooses not to exceed regardless of risk aversion, simply to avoid reducing returns with trading costs. In this context, margin requirements based on volatility (such as value at risk and its variations) are binding only when they reduce leverage below the multiplier, and are otherwise redundant. In addition, the multiplier shows that an exogenous increase in trading costs, such as a proportional Tobin tax on financial transactions, implicitly reduces the maximum leverage that any investor who seeks return is willing to take, regardless of risk attitudes.

Second, two assets with the same Sharpe ratio do not generate the same efficient frontier with trading costs, and more volatility leads to a superior frontier. For example (Table \ref{tab:mult}) with a 1\% spread, the maximum leveraged return on an asset with 10\% volatility and 5\% return is $7.72\times 5\% \approx 39\%$. By contrast,  an asset with 50\% volatility and 25\% return (equivalent to the previous one from a classical viewpoint, as it has the same Sharpe ratio 0.5), leads to a maximum leveraged return of $3.66\times 25\% \approx 92\%$. The reason is that a more volatile asset requires a lower leverage ratio (hence lower rebalancing costs) to reach a certain return. Thus, an asset with higher volatility spans an efficient frontier that achieves higher returns through lower costs. 

This paper bears on the established literature on portfolio choice with frictions. The effect of transaction costs on portfolio choice is first studied by \citet*{MR0469196}, \citet*{constantinides.86}, and \citet*{MR1080472}, who identify a wide no-trade region, and derive the optimal trading boundaries through numerical procedures. While these papers focus on the maximization of expected utility from intertemporal consumption on an infinite horizon, \citet*{MR942619}, and  \citet*{dumas.luciano.91} show that similar strategies are obtained in a model with terminal wealth and a long horizon -- time preference has negligible effects on trading policies. This paper adopts the same approach of a long horizon, both for the sake of tractability, and because it focuses on the trade-off between return, risk, and costs, rather than consumption. 

Our asymptotic results for positive risk aversion are similar in spirit to the ones derived by  \citet*{MR1284980}, \citet*{MR2076549}, \citet*{gerhold.al.11}, and \citet*{kallsen2013general}, whereby transaction costs imply a no-trade region with width of order $O(\eps^{1/3})$ and welfare costs of order $O(\eps^{2/3})$. We also find that the trading boundaries obtained from a local mean-variance criterion are equivalent at the first order to the ones obtained from power utility.

The risk-neutral expansions and the limits of leverage of order $O(\eps^{-1/2})$ are new, and are qualitatively different from the risk-averse case. These results are not regular perturbations of a frictionless analogue, which is ill-posed. They are rather singular perturbations, which display the speed at which the frictionless problem becomes ill-posed as the crucial friction parameter vanishes.

Finally, this paper connects to the recent work of \cite{frazzini2012embedded} on embedded leverage. If different investors face different leverage constraints, they find that in equilibrium assets with higher factor exposures trade at a premium, thereby earning a lower return. \cite{frazzini2014betting} confirm this prediction across a range of markets and asset classes, and \cite{asness2012leverage} use it to explain the performance risk-parity strategies. With exogenous asset prices, we find that assets with higher volatility generate a superior efficient frontier by requiring lower rebalancing costs for the same return. This observation suggests that the embedded leverage premium may be induced by rebalancing costs in addition to leverage constraints, and should be higher for more illiquid assets.

 The paper is organized as follows: section \ref{sec: the model} introduces the model and the optimization problem. Section \ref{sec:mainres} contains the main results, which characterize the efficient frontier in the risk-averse (Theorem \ref{thm free b}) and risk-neutral (Theorem \ref{th multiplier}) cases. Section \ref{sec: implications} discusses the implications of these results for the efficient frontier, the trading boundaries of optimal policies and the embedded leverage effect. The section includes two supporting results, which show that the risk-neutral solutions arise as limits of their risk-averse counterparts for low risk-aversion (Theorem \ref{lem convergence}), and that the risk-neutral solutions are not constrained by the solvency condition (section \ref{sec tb}). Section \ref{sec: Heuristics} offers a derivation of the main free-boundary problems from heuristic control arguments, and concluding remarks are in section \ref{sec: conclusion}. All proofs are in the appendix.

\section{Model}\label{sec: the model}

The market includes one safe asset earning a constant interest rate of $r\ge 0$ and a risky asset with ask (buying) price $S_t$ that follows
\[
\frac{dS_t}{S_t}=(\mu+r)dt+\sigma dB_t, \quad S_0, \sigma,\mu > 0,
\]
where $B$ is a standard Brownian motion. The risky asset's bid (selling) price is $(1-\varepsilon)S_t$, which implies a constant relative bid-ask spread of $\varepsilon>0$, or, equivalently, constant proportional transaction costs.
\nada{
A self-financing trading strategy is summarized by its initial capital $x$ and the number of shares $\varphi_t$ of the risky asset held at time $t$. Denote by $w_t$ the fund's wealth at time $t$, which is the sum of the safe position
$x - \int_0^t S_s d\varphi_s - \varepsilon\int_0^t S_s d\varphi^\downarrow_s$ and the risky position $S_t \varphi_t $ evaluated at the ask price\footnote{The convention of evaluating the risky position at the ask price is inconsequential. Using the bid price instead leads to the same results up to a change of notation.}:
\begin{equation}\label{eq:self financing}
w_t = x - \int_0^t S_s d\varphi_s - \varepsilon\int_0^t S_s d\varphi^\downarrow_s
+S_t \varphi_t.
\end{equation}
We further require a strategy $\varphi$ to be solvent, in that its corresponding wealth $w_t$ is strictly positive at all times. (Admissible strategies are formally described in Definition \ref{def: admiss} below.)
}

We investigate the trade-off between a portfolio's average return against its realized variance. Denoting by $w_t$ the portfolio value at time $t$, for an investor who observes returns with frequency $\Delta t=T/n$ in the time-interval $[0,T]$, the average return and its continuous-time approximation are\footnote{All discrete statistics on this section converge in probability to their continuous-time counterparts. The budget equation and the definition of admissible strategies are in appendix \ref{sec: prereq} below.} 
\[
\frac{1}{n \Delta t}\sum_{k=1}^{n} \left(\frac{w_{k\Delta t}}{w_{(k-1)\Delta t}}-1\right)
\approx
\frac1T\int_0^T \frac{dw_t}{w_t}.
\]
In the familiar setting of no trading costs, $\frac1T\int_0^T \frac{dw_t}{w_t} = r+ \frac1 T\int_0^T \mu\pi_t dt+\frac{1}{T}\int_0^T \sigma\pi_t dB_t$, where $\pi_t$ is the portfolio weight of the risky asset, hence the average return equals the average risky exposure times its excess return, plus the safe rate.

Likewise, the average squared volatility on $[0,T]$ is obtained by the usual variance estimator applied to returns, and has the continuous-time approximation 
\[
\frac{1}{n \Delta t}
\sum_{k=1}^{n}
\left(\frac{w_{k\Delta t}}{w_{(k-1)\Delta t}}-1\right)^2
\approx 
\frac1T\int_0^T \frac{d\langle w\rangle_t}{w_t^2}
\]
 reducing to $\frac{\sigma^2} T\int_0^T \pi^2_t dt$ in the absence of trading costs.

With these definitions, the mean-variance trade-off is captured by maximizing 
\begin{equation}\label{eq: obj abstract}
\frac1 T \mathbb E\left[ \int_0^T \frac{dw_t}{w_t}-\frac\gamma 2 
\left\langle \int_0^\cdot 
\frac {dw_t}{w_t}\right\rangle_T
\right] ,
\end{equation}
where the parameter $\gamma>0$ is interpreted as a proxy for risk-aversion. 

This objective nests several familiar problems. Without trading costs it reduces to
\begin{equation}\label{eq: obj frictionless}
\frac1 T \mathbb E\left[ \int_0^T \left(\mu \pi_t -\frac\gamma 2 \sigma^2\pi_t^2\right) dt \right]
\end{equation}
which is maximized by the optimal constant-proportion portfolio $\pi = \frac{\mu}{\gamma \sigma^2}$ dating back to Markowitz and Merton, and confirms that in a geometric Brownian motion market with costless trading, the objective considered here is equivalent to utility-maximization with constant relative risk aversion. 
With or without transaction costs, the risk-neutral objective $\gamma = 0$ boils down to the average annualized return over a long horizon, while $\gamma=1$ reduces to logarithmic utility.

\nada{
To proceed further, note first that the objective function \eqref{eq: obj abstract} has a more concrete expression,
see Lemma \ref{le: rewriting obj fun}: For any $T>0$ and for any admissible trading strategy $\varphi$, 
\begin{align}
F_T(\varphi) :=&
\frac1 T \mathbb E\left[ \int_0^T \frac{dw_t}{w_t}-\frac\gamma 2 
\left\langle \int_0^\cdot 
\frac {dw_t}{w_t} \right\rangle_T
\right]\\
=&
\label{altinto}
r+\frac{1}{T}\mathbb E\left[\int_0^T \left(\mu \pi_t-\frac{\gamma \sigma^2}{2}\pi_t^2\right)dt-\varepsilon\int_0^T\pi_t\frac{d\varphi^\downarrow_t}{\varphi_t} \right].
\end{align}
}
Trading costs make \eqref{eq: obj abstract} lower than \eqref{eq: obj frictionless}, as they hinder continuous portfolio rebalancing and make constant-proportion strategies unfeasible. The reason is that it is costly to keep the exposure to the risky asset high enough to achieve the desired return, and low enough to limit the level of risk -- trading costs reduce returns and increase risk.

To neglect the spurious, non-recurring effects of portfolio set-up and liquidation, we focus on the Equivalent Safe Rate\footnote{In this equation the $\limsup$ is used merely to guarantee a good definition a priori. A posteriori, optimal strategies exist in which the limit superior is a limit, hence the similar problem defined in terms of $\liminf$ leads to the same solution.}
\begin{align}\label{eq: obj asympt}
\esr &:= \limsup_{T\rightarrow\infty} \frac1 T \mathbb E\left[ \int_0^T \frac{dw_t}{w_t}-\frac\gamma 2 
\left\langle \int_0^\cdot 
\frac {dw_t}{w_t}\right\rangle_T
\right]
\end{align}
which is akin to the one used by \cite{dumas.luciano.91} in the context of utility maximization.

\section{Main Results}\label{sec:mainres}

\subsection{Risk aversion and efficient frontier}
The first result characterizes the optimal solution to the main objective in \eqref{eq: obj asympt} in the usual case of a positive aversion to risk ($\gamma>0$). In this setting, the next theorem shows that trading costs create a no-trade region around the frictionless portfolio $\pi_* = \frac\mu{\gamma \sigma^2}$, and states the asymptotic expansions of the resulting average return and standard deviation\footnote{The exact formulae for average return, standard deviation, and average trading costs are in Appendix \ref{app proof main theorem}.}, thereby extending the familiar efficient frontier to account for trading costs.

\begin{theorem}\label{thm free b}
\noindent
Let $\frac{\mu}{\gamma\sigma^2}\neq 1$.
\begin{enumerate}
\item \label{thm free b part 1}
For any $\gamma>0$ there exists $\varepsilon_0>0$ such that for all $\varepsilon<\varepsilon_0$, 
there is a unique solution $(W,\zeta_-,\zeta_+)$, with $\zeta_-<\zeta_+$, for the free boundary problem 
\begin{align}\label{eq: TKA fbp}
&\textstyle{\frac{1}{2}\sigma^2 \zeta^2 W''(\zeta)+(\sigma^2+\mu)\zeta W'(\zeta)+\mu W(\zeta)-\frac{1}{(1+\zeta)^2}\left(\mu-\gamma\sigma^2\frac{\zeta}{1+\zeta}\right)=0,}\\\label{initial0 TKA}
&\textstyle{W(\zeta_-)=0}\\\label{initial1 TKA}
&\textstyle{W'(\zeta_-)=0,}\\\label{terminal0 TKA}
&\textstyle{W(\zeta_+)=\frac{\varepsilon}{(1+\zeta_+)(1+(1-\varepsilon)\zeta_+)},}\\
\label{terminal1 TKA}
&\textstyle{W'(\zeta_+)=\frac{\varepsilon(\varepsilon-2(1-\varepsilon)\zeta_+-2)}{(1+\zeta_+)^2(1+(1-\varepsilon)\zeta_+)^2}}
\end{align}

\item \label{thm free b part 2}
The trading strategy that buys at $\pi_- := \zeta_-/(1+\zeta_-)$ and sells at $\pi_+ := \zeta_+/(1+\zeta_+)$ as little as to keep the risky weight $\pi_t$ within the interval $[\pi_-,\pi_+]$ is optimal.

\item \label{thm free b part 3}
The maximum performance is 
\begin{equation}
\textstyle{\max_{\varphi\in\Phi}\lim_{T\rightarrow\infty}
\frac1T 
\mathbb E\left[
\int_0^T \left( \mu \pi_t - \frac{\gamma\sigma^2}{2} \pi_t^2 \right)dt
-{\varepsilon}\int_0^T\pi_t\frac{d\varphi^\downarrow_t}{\varphi_t} 
\right]
= \mu \pi_- -\frac{\gamma \sigma^2}{2}\pi_-^2,}
\end{equation}
where $\Phi$ is the set of admissible strategies in Definition \ref{def: admiss} below, $\varphi_t = \pi_t w_t/S_t$ is the number of shares held at time $t$, and $\varphi^\downarrow_t$ is the cumulative number of shares sold up to time $t$.

\item\label{thm free b part 4} 
The trading boundaries $\pi_-$ and $\pi_+$ have the asymptotic expansions
\begin{equation}\label{eq: trading boundaries TKA x}
\textstyle{\pi_\pm=\pi_*\pm\left(\frac{3}{4\gamma}\pi_*^2(\pi_*-1)^2\right)^{1/3}\varepsilon^{1/3}-\frac{(1-\gamma)\pi_*}{\gamma}\left(\frac{\gamma \pi_*(\pi_*-1)}{6}\right)^{1/3}\varepsilon^{2/3} +O(\varepsilon).}
\end{equation}

The long-run mean ($\hat m $), standard deviation ($\hat s$), Sharpe ratio ($(\hat m-r)\hat s$), average trading costs ($\avtrco$) and equivalent safe rate ($\esr$) have expansions\footnote{We are using the convention $a^{1/n}=\text{sign}(a)\,\vert a \vert^{1/n}$ for any $a\in\mathbb R$ and odd integer $n$, and $a^{2/n}=(a^2)^{1/n}$.}
\begin{align}
\label{eq: mean asympt}
\textstyle{\hat m}&:=\textstyle{\lim_{T\rightarrow\infty}\frac{1}{T}\int_0^T \frac{dw_t}{w_t}=r+\frac{\mu^2}{\gamma \sigma^2} -\frac{\sigma^2\pi_*( 5 \pi_*-3)}{2}\left(\frac{\gamma\pi_*(\pi_*-1)}{6}\right)^{1/3}\varepsilon^{2/3}+O(\varepsilon),}\\
\label{eq: vol asympt}
\textstyle{\hat s}&:= \textstyle{\lim_{T\rightarrow\infty}\sqrt{\frac{1}{T}\left\langle\int_0^\cdot\frac{dw_t}{w_t}\right\rangle_T}=\frac{\mu}{\gamma \sigma}-\frac{\sigma (7\pi_*-3)}{4\gamma}\left(\frac{\gamma\pi_*(\pi_*-1)}{6}\right)^{1/3}\varepsilon^{2/3}+O(\varepsilon),}\\
\label{eq: Sharpe asympt}
\textup{SR} &:=\frac{\hat m-r}{\hat s}=\frac{\mu}{\sigma}+\frac{3}{4\cdot {6}^{1/3}} (\pi_*-1)(\gamma \pi_*(1-\pi_*))^{1/3} \varepsilon^{2/3}+O(\varepsilon)\\
\label{eq: trading costs asympt}
\textstyle{\avtrco}&:=\textstyle{\lim_{T\rightarrow\infty}\frac{1}{T}\int_0^T\pi_t\frac{d\varphi^\downarrow_t}{\varphi_t}=\frac{3\sigma^2}{\gamma} \left(\frac{\gamma \pi_*(\pi_*-1)}{6}\right)^{4/3}\varepsilon^{2/3}+O(\varepsilon),}\\
\label{eq: performance asympt}
\textstyle{\esr} &=\textstyle{r+\frac{\gamma \sigma^2}{2}\pi_*^2-\frac{\gamma \sigma^2}{2}\left(\frac{3}{4 \gamma}\pi_*^2(\pi_*-1)^2\right)^{2/3}\varepsilon^{2/3} + O(\varepsilon).}
\end{align}

\end{enumerate}
\end{theorem}
\begin{proof}
The proof of the main part of this theorem is divided into Propositions \ref{prop1}, \ref{prop2} and \ref{prop3} in Appendix \ref{appendix HJB}. The proof of the asymptotic results is in section \ref{as of th31}.
\end{proof}

\subsection{Risk neutrality and limits of leverage}

In contrast to the risk-averse objective considered above, the risk-neutral objective leads to a solution which does not have a frictionless analogue: for small trading costs, both the optimal policy and its performance become unbounded as the optimal leverage increases arbitrarily. The next result describes the solution to the risk-neutral problem, identifying the approximate dependence of the leverage multiplier and its performance on the asset's risk, return and liquidity.

\begin{theorem}\label{th multiplier}
Let $\gamma=0$.
\begin{enumerate}
\item  \label{th multiplier issue 1}
There exists $\varepsilon_0>0$ such that for all $\varepsilon<\varepsilon_0$, the free boundary problem  \eqref{eq: TKA fbp}--\eqref{terminal1 TKA} has a unique solution $(W,\zeta_-,\zeta_+)$ with $\zeta_-<\zeta_+$. 

\item \label{th multiplier issue 2}

The trading strategy $\hat\varphi$ that buys at $\pi_- := \zeta_-/(1+\zeta_-)$ and sells at $\pi_+ := \zeta_+/(1+\zeta_+)$ as little as to keep the risky weight $\pi_t$ within the interval $[\pi_-,\pi_+]$ is optimal.

\item \label{th multiplier issue 3}
The maximum expected return is 
\begin{equation}
\textstyle{\max_{\varphi\in\Phi}\lim_{T\rightarrow\infty}
\frac1 T \int_0^T \frac{dw_t}{w_t}
= r + \mu \pi_- .}
\end{equation}

\item  \label{th multiplier issue 4}
The trading boundaries have the series expansions
\begin{align}\label{as multiplier piminus}
\textstyle{\pi_- =}&\textstyle{ (1-\kappa)\kappa^{1/2}\left(\frac{\mu}{\sigma^2}\right)^{1/2}\varepsilon^{-1/2} + 1 + O(\varepsilon^{1/2}),}\\
\textstyle{\pi_+ =}& \textstyle{\kappa^{1/2}\left(\frac{\mu}{\sigma^2}\right)^{1/2}\varepsilon^{-1/2} + 1 + O(\varepsilon^{1/2}),}
\end{align}
where $\kappa\approx 0.5828$ is the unique solution to
\begin{equation}\label{unicorn dimless}
\textstyle{f(\xi):=\frac{3}{2}\xi+\log(1-\xi)=0,\quad \xi\in(0,1).}
\end{equation}
\end{enumerate}
\end{theorem}
\begin{proof}
See Appendix \ref{sec: proof th multiplier} below.
\end{proof}

The next section discusses how these results modify the familiar intuition about risk, return, and performance evaluation in the context of trading costs.

\section{Implications and Applications}\label{sec: implications}

\subsection{Efficient frontier}

Theorem \ref{thm free b} extends the familiar efficient frontier to account for trading costs. Compared to the linear frictionless frontier, average returns decline because of rebalancing losses. Average volatility increases because more risk becomes necessary to obtain a given return net of trading costs. 

To better understand the effect of trading costs on return and volatility, consider the dynamics of the portfolio weight in the absence of trading, which is
\begin{equation}
d\pi_t = \pi_t (1-\pi_t) (\mu - \sigma^2 \pi_t)dt + \sigma \pi_t (1-\pi_t) dB_t
.
\end{equation}
The central quantity here is the portfolio weight volatility $\sigma \pi_t (1-\pi_t)$, which vanishes for the single-asset portfolios $\pi_t = 0$ or $\pi_t = 1$, remains bounded above by $\sigma/4$ in the long-only case $\pi_t \in [0,1]$, and rises quickly with leverage ($\pi_t>1$). This quantity is important because it measures the extent to which a portfolio, left to itself, strays from its initial composition in response to market shocks and, by reflection, the quantity of trading that is necessary to keep it within some region. In the long-only case, the portfolio weight volatility decreases as the no-trade region widens to span $[0,1]$, which means that a portfolio tends to spend more time near the boundaries. By contrast, with leverage portfolio weight volatility increases, which means that a wider boundary does not necessarily mitigate trading costs.

\begin{figure}[t]
\centering
\includegraphics[width=\textwidth]{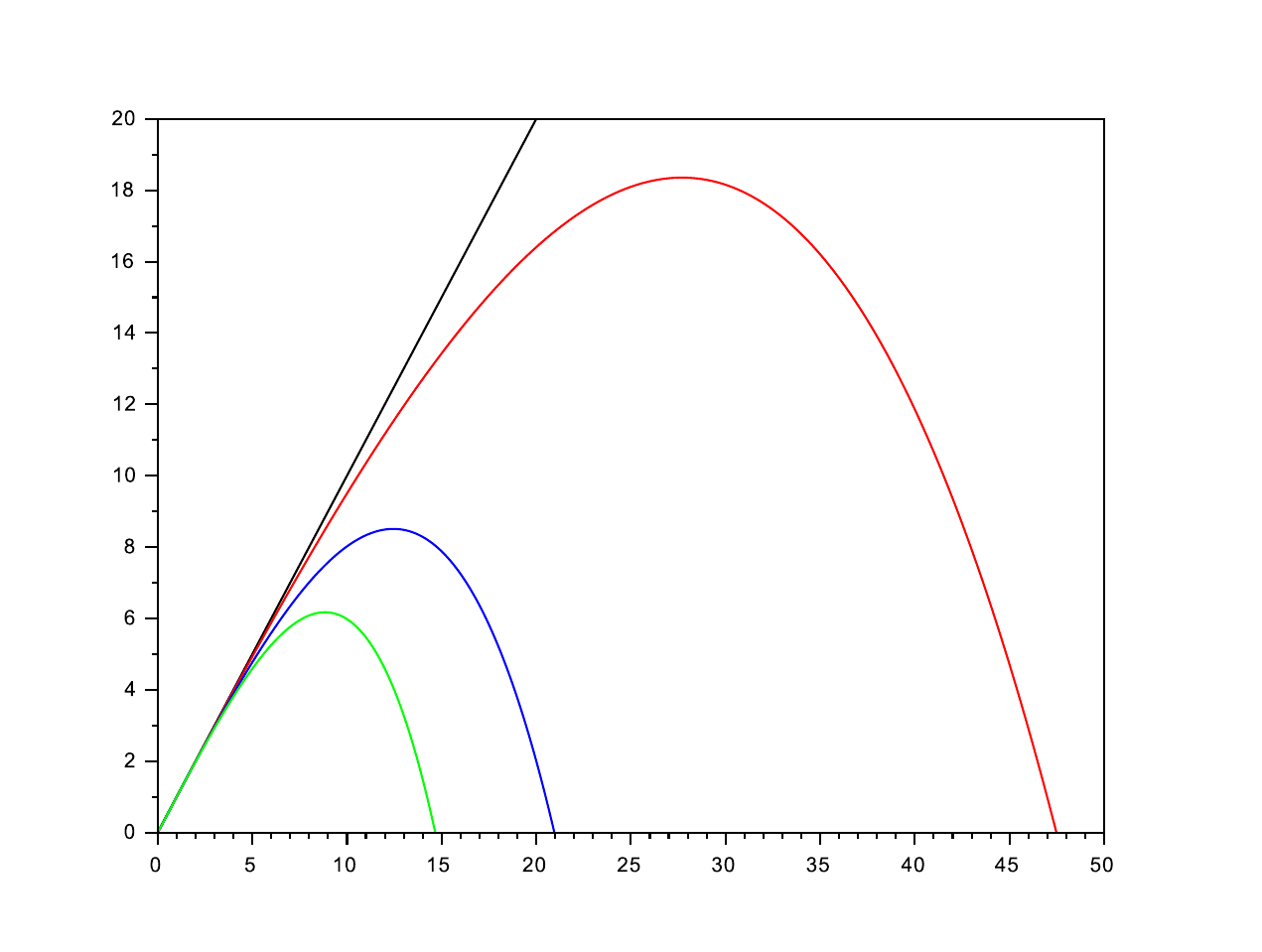}
\caption{\label{fig:mvt multiple}
Efficient Frontier with trading costs, as expected excess return (vertical axis, in multiples of the asset's expected excess return) against standard deviation (horizontal axis, in multiples of the asset's volatility). The asset has expected excess return $\mu=8\%$, volatility $\sigma=16 \%$, and bid-ask spread of $0.1 \%, 0.5 \%, 1 \%$. The upper line is the frictionless efficient frontier. The maximum of each curve is the leverage multiplier.}
\end{figure}

Consistent with this intuition, equation \eqref{eq: mean asympt} shows that the impact of trading costs is small on long-only portfolios, but rises quickly with leverage, reducing returns at the order of $\varepsilon^{2/3}$ for $\pi_*>1$. Of course, this expansion is valid for small $\varepsilon$ while holding the value of $\gamma$ fixed. As $\gamma$ declines to zero, both the expected return and volatility diverge, but so does the impact of trading costs, making the asymptotics for $\gamma>0$ uninformative for the risk-neutral limit $\gamma=0$.

The performance \eqref{eq: performance asympt} coincides at the first order with the equivalent safe rate from utility maximization with constant relative risk aversion $\gamma$ \cite[Equation (2.4)]{gerhold.al.11}, supporting the interpretation of $\gamma$ as a risk-aversion parameter, and confirming that, for asymptotically small costs, the efficient frontier captures the risk-return trade-off faced by a utility maximizer.

Figure \ref{fig:mvt multiple} displays the effect of trading costs on the efficient frontier. As the bid-ask spread declines, the frontier increases to the linear frictionless frontier, and the asymptotic results in the theorem become more accurate. However, if the spread is held constant as leverage (hence volatility) increases, the asymptotic expansions become inaccurate, and in fact the efficient frontier ceases to increase at all after the leverage multiplier is reached.

\subsection{Trading boundaries}\label{sec tb}
\begin{figure}[t]
\centering
\includegraphics[width=\textwidth]{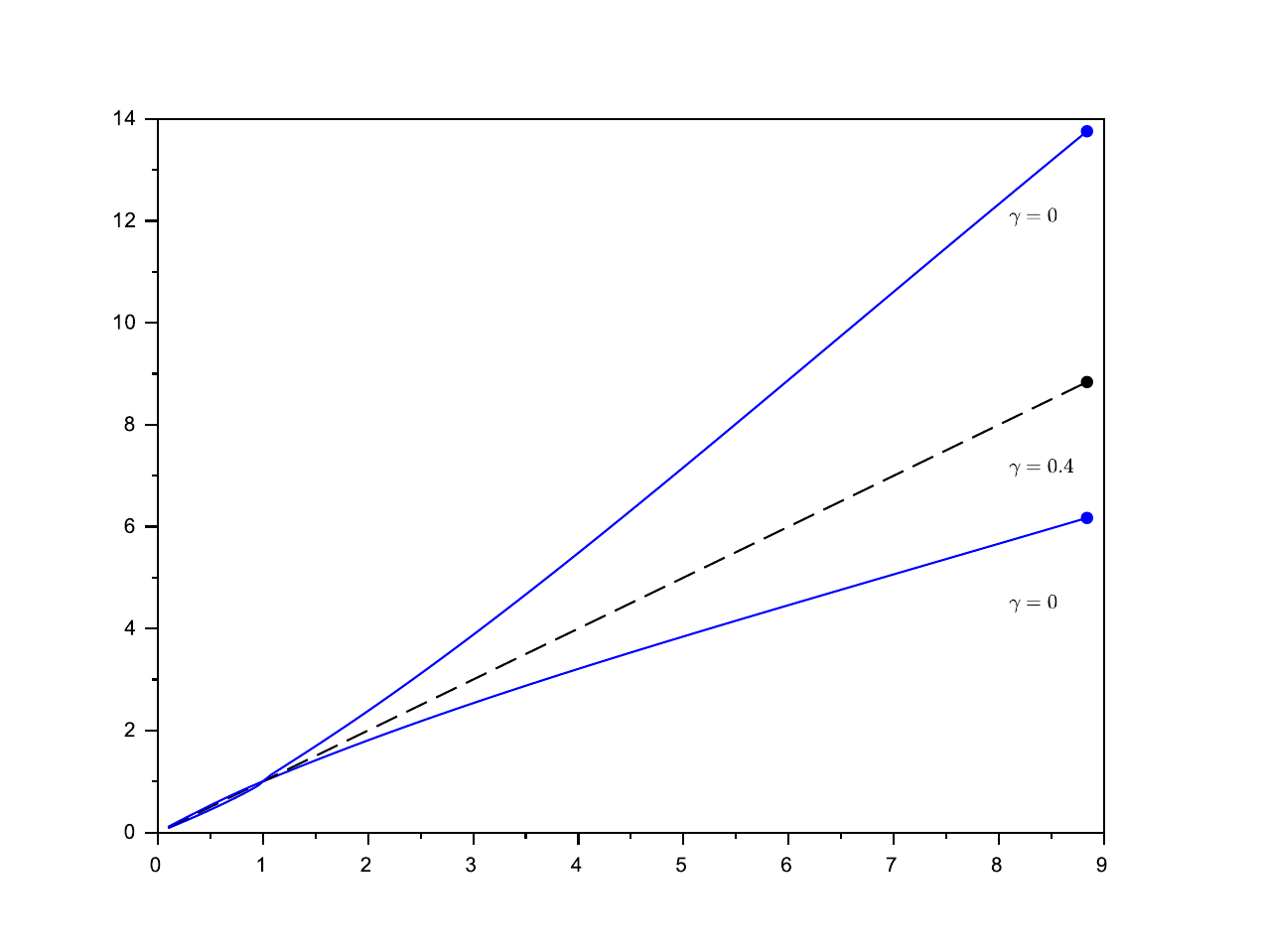}
\caption{Trading boundaries $\pi_\pm$ (vertical axis, outer curves, as risky weights) and implied Merton fraction (middle curve) against average portfolio volatility (horizontal axis, as multiples of $\sigma$). $\mu=8\%$, $\sigma=16\%$, and $\varepsilon=1\%$.}
\label{fig:boundaries}
\end{figure}

Each point in the efficient frontier corresponds to a rebalancing strategy that is optimal for some value of the risk-aversion parameter $\gamma$. For small trading costs, equation \eqref{eq: trading boundaries TKA x} implies that the trading boundaries corresponding to the efficient frontier depart from the ones arising in utility maximization, which are \citep{gerhold.al.11}
\begin{equation}
\pi_\pm=\pi_*\pm\left(\frac{3}{4\gamma}\pi_*^2(1-\pi_*)^2\right)^{1/3}\varepsilon^{1/3} +O(\varepsilon).
\end{equation}
The term of order $\varepsilon^{2/3}$ vanishes for $\gamma = 1$ because this case coincides with the maximization of logarithmic utility. For high levels of leverage ($\gamma <1$ and $\pi_*>1$), this term implies that the trading boundaries that generate the efficient frontier are lower than the trading boundaries that maximize utility. In Figure \ref{fig:boundaries}, $\gamma \rightarrow \infty$ corresponds to the safe portfolio in the origin (0,0), while $\gamma = \mu/\sigma^2$ to the risky investment (1,1), which has by definition the same volatility and return as the risky asset. As $\gamma$ declines to zero, the trading boundaries converge to the right endpoints, which correspond to the strategy that maximizes average return with no regard for risk, thereby achieving the multiplier.

 As leverage increases, the sell boundary rises more quickly than the buy boundary (Figure \ref{fig:boundaries}). For example, the risk-neutral portfolio tolerates leverage fluctuations from approximately $6$ to $14$. The locations of these boundaries trade off the need to keep exposure to the risky asset high to maximize return while also keeping rebalancing costs low. Risk aversion makes boundaries closer to each other by penalizing the high realized variance generated by the wide risk-neutral boundaries.

Importantly, these boundaries remain finite even as the frictionless Merton portfolio $\mu/(\gamma\sigma^2)$ diverges to infinity with $\gamma$ declining to zero. Thus the no-trade region is not symmetric around the frictionless portfolio, in contrast to the boundaries arising from utility maximization \citep{gerhold.al.11}, which are always symmetric, and hence diverge when $\gamma$ is low. 
The difference is that here the risk-neutral objective is to maximize the expected \emph{return} of the portfolio, while a risk-neutral utility maximizer focuses on expected \emph{wealth}. In a  frictionless setting this distinction is irrelevant, and an investor can use a return-maximizing policy to maximize wealth instead. But trading costs drive a wedge between these two ostensibly equivalent risk-neutral criteria -- maximizing expected return is not the same as maximizing expected wealth. 

In the risk-neutral case (Theorem \ref{th multiplier} \ref{th multiplier issue 4}) the optimal trading boundaries satisfy the approximate relation
\begin{equation}
\frac{\pi_-}{\pi_+} \approx 0.4172 
\end{equation}
which is universal in that it holds for any asset, regardless of risk, return and liquidity. This relation means that an optimal risk-neutral rebalancing strategy should always tolerate wide variations in leverage over time, and that the maximum allowed leverage should be approximately 2.5 times the minimum. More frequent rebalancing cannot achieve the maximum return: it can be explained either by risk aversion or by elements that lie outside the model, such as price jumps.

Finally, note that the solvency constraint that wealth remain positive at all times implies that\footnote{ \label{foot special} Any potentially optimal strategy has positive exposure ($\pi_t> 0$), as the asset price has a positive risk premium (Remark \ref{rem: range of trade} below). Denoting by $X_t=w_t-\varphi_t S_t$ the safe position at time $t\geq 0$, where $w_t$ is total portfolio wealth, the liquidation value is $w_t-\varepsilon \varphi_t S_t\geq 0$, which implies $1-\varepsilon \frac{\varphi_t S_t}{w_t}>0$ and thus the claim.} $\pi_t< \frac{1}{\varepsilon}$ for every admissible trading strategy. As $\pi_t\leq \pi_+$ for the optimal trading strategy in Theorem \ref{thm free b} and Theorem \ref{th multiplier}, the upper bound $\pi_t\le \frac1\varepsilon$ is never binding for realistic bid-ask-spreads.

\subsection{Embedded leverage}

\begin{figure}[t]
\centering
\includegraphics[width=\textwidth]{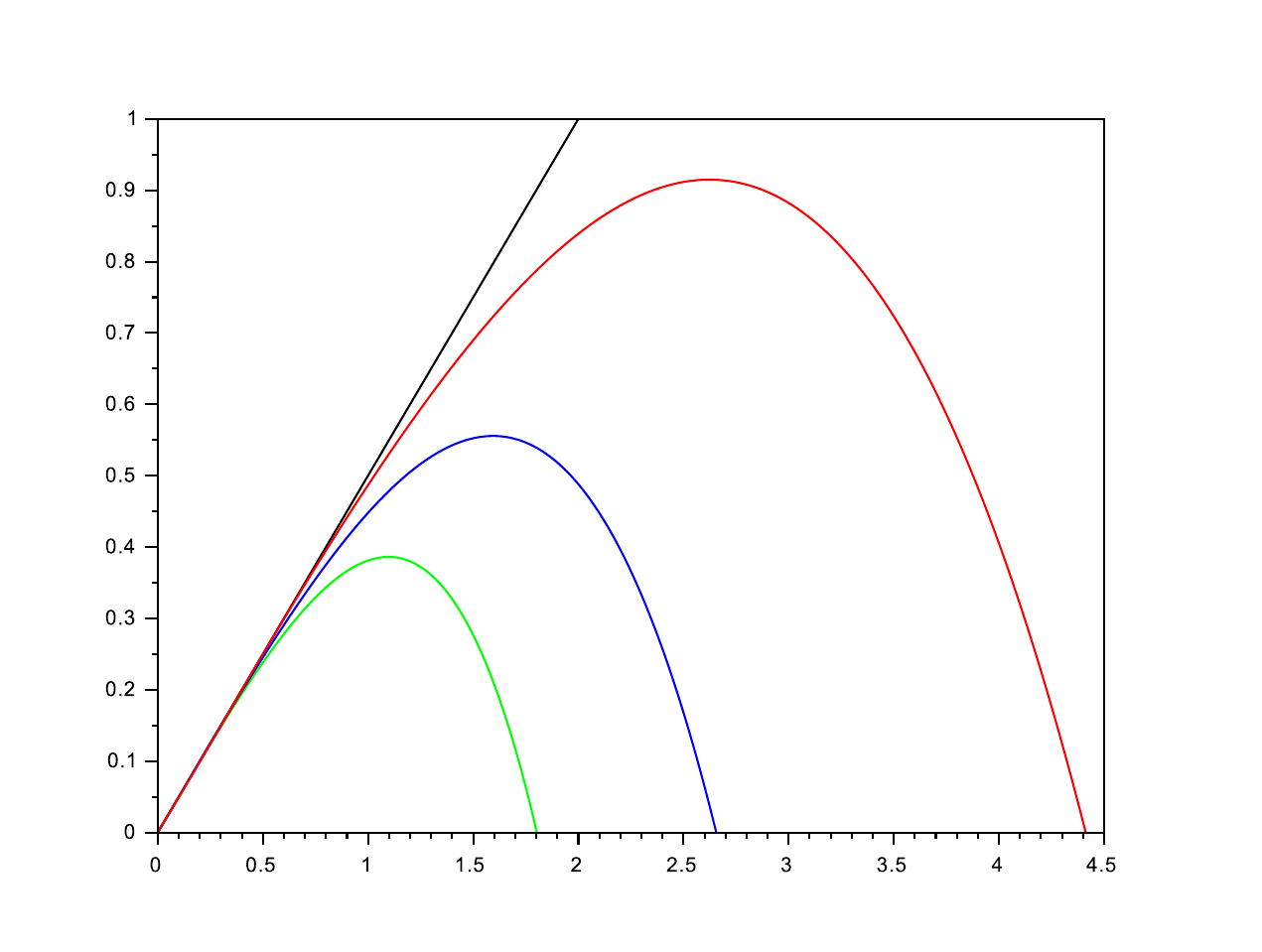}
\caption{Efficient Frontier, as average expected excess return (vertical axis) against volatility (horizontal axis), for an asset with Sharpe ratio $\mu/\sigma=0.5$, for various levels of asset volatility, from $10\%$ (bottom), $20\%$, to $50\%$ (top), for a bid-ask spread $\varepsilon = 1\%$. The straight line is the frictionless frontier. }
   \label{fig:mvt constant sharpe}
\end{figure}
In frictionless markets, two perfectly correlated assets with equal Sharpe ratio generate the same efficient frontier, and in fact the same payoff space. This equivalence fails in the presence of trading costs: as the more volatile asset has a proportionally higher return, it can be traded to generate higher returns with lower leverage ratios, resulting in an efficient frontier that dominates (for high returns) the one generated by the less volatile asset. Figure \ref{fig:mvt constant sharpe} (top of the three curves) displays this phenomenon: for example, a portfolio with an average return of 50\% net of trading costs is obtained from an asset with 25\% return and 50\% volatility at a small cost, as an average leverage factor of 2 entails moderate rebalancing.

Achieving the same 50\% return from an asset with 20\% volatility (and 10\% return) is more onerous: trading costs require leverage higher than 5, which in turn increases trading costs. Overall, the resulting portfolio needs about 120\% rather than 100\% volatility 
to achieve the desired 50\% average return (middle curve in Figure \ref{fig:mvt constant sharpe}).

From an asset with 10\% volatility (and 5\% return), obtaining a 50\% return net of trading costs is impossible (bottom curve in Figure \ref{fig:mvt constant sharpe}), because the leverage multiplier is less than 8 (Table \ref{tab:mult}, top right), and therefore the return can be scaled to less than 40\%. The intuition is clear: increasing leverage also increases trading costs, which in turn call for more leverage to increase return, but also increase costs. At some point, the marginal net return from more leverage becomes zero, and further increases are detrimental.

Because an asset with higher volatility is superior to another one, with equal Sharpe ratio and perfectly correlated, but with lower volatility, the model suggests that in equilibrium they cannot coexist, and that the asset with lower volatility should offer a higher return to be held by investors. Indeed, \cite{frazzini2012embedded,frazzini2014betting} document significant negative excess returns in assets with embedded leverage (higher volatility), and offer a theoretical explanation based on heterogeneous leverage constraints, which lead more constrained investors to bid up prices (and hence lower returns) of more volatile assets. Our results hint that the same phenomenon may arise even in the absence of constraints, as a result of rebalancing costs. In contrast to constraints-based explanations, our model suggests that the premium for embedded leverage should be higher for more illiquid assets.

\subsection{From risk aversion to risk neutrality}\label{sec convergence}
Theorems \ref{thm free b} and \ref{th multiplier} are qualitatively different: while Theorem \ref{thm free b} with positive risk aversion leads to a regular perturbation of the Markowitz-Merton solution, Theorem \ref{th multiplier} with risk-neutrality leads to a novel result with no meaningful analogue  in the frictionless setting -- a singular perturbation. Furthermore, a close reading of the statement of Theorem \ref{thm free b} shows that the existence of a solution to the free-boundary problem, and the asymptotic expansions, hold for $\varepsilon$ less than some threshold $\bar\varepsilon(\gamma)$ that depends on the risk aversion $\gamma$. In particular, if $\gamma$ approaches zero while $\varepsilon$ is held constant, Theorem \ref{thm free b} does not offer any conclusions on the convergence of the risk-averse to the risk-neutral solution. Still, if the risk-neutral result is to be accepted as a genuine phenomenon rather than an artifact, it should be clarified whether the risk averse trading policy and its performance converge to their risk neutral counterparts as risk aversion vanishes. The next result resolves this point under some parametric restrictions. 
Denote by
\[
\textstyle{G(\zeta):=\frac{\varepsilon}{(1+\zeta)(1+(1-\varepsilon)\zeta)},\quad h(\zeta)=\mu\left(\frac{\zeta}{1+\zeta}\right)-\frac{\gamma\sigma^2}{2}\left(\frac{\zeta}{1+\zeta}\right)^2}
\]
and associate to any solution  $(W(\cdot;\gamma),\zeta_-(\gamma),\zeta_+(\gamma))$ of the free boundary problem \eqref{eq: TKA fbp} the function 
\[
\textstyle{\hat W(\zeta;\gamma):=\begin{cases} 0,\quad \quad&\zeta<\zeta_-(\gamma)\\ W(\zeta;\gamma),\quad &\zeta\in [\zeta_-(\gamma),\zeta_+(\gamma)]\\ G(\zeta),\quad\quad& \zeta\geq \zeta_+(\gamma)\end{cases},}
\]
which naturally extends $W$ to the left and right of the free-boundaries.

\begin{theorem}\label{lem convergence}
Let $\mu>\sigma^2$, $\bar\varepsilon>0$, and $\bar \gamma>0$, and assume that for any $\gamma\in [0,\bar\gamma]$ the free boundary problem \eqref{eq: TKA fbp} has a unique solution $(W,\zeta_-,\zeta_+)$ satisfying $\zeta_+<-1/(1-\varepsilon)$ and that the function $\hat W$
satisfies, for each  $\gamma\in (0,\bar\gamma]$, the HJB equation
\begin{equation}\label{HJB chapter 5x}
\textstyle{\min\left(\frac{\sigma^2}{2}\zeta^2 \hat W'+\mu \zeta  \hat W- h(\zeta)+h(\zeta_-), G(\zeta)-\hat W,\hat W\right)=0.}
\end{equation}
Then, \eqref{HJB chapter 5x} is satisfied also for $\gamma=0$, and for each $\gamma\in [0,\bar \gamma]$, the trading strategy that buys at $\pi_- (\gamma)=\frac{\zeta_-(\gamma)}{1+\zeta_-(\gamma)}$
and sells at $\pi_+ (\gamma)=\frac{\zeta_+(\gamma)}{1+\zeta_+(\gamma)}$ to keep the risky weight $\pi_t$ within the interval $[\pi_-(\gamma),\pi_+(\gamma)]$ is optimal.
Furthermore, $\zeta_\pm(\gamma)\rightarrow \zeta_\pm(0)$ and $\hat W(\zeta;\gamma)\rightarrow \hat W(\zeta;0)$ as $\gamma \downarrow 0$, each $\zeta\in\mathbb R$.
\end{theorem}

In summary, this result confirms that, as the risk-aversion parameter $\gamma$ declines to zero, the risk-averse policy in Theorem \ref{thm free b} converges to the risk-neutral policy in Theorem \ref{th multiplier}, and that the corresponding mean-variance objective in Theorem \ref{thm free b} converges to the average return in Theorem \ref{th multiplier}.

\section{Heuristic Solution}\label{sec: Heuristics}
This section offers a heuristic derivation of the HJB equation. Let $(\varphi_t^\uparrow)_{t\ge 0}$ and $(\varphi_t^\downarrow)_{t\ge 0}$ denote the cumulative number of shares bought and sold, respectively. The finite-horizon objective \eqref{eq: obj abstract} reduces to the expression (compare eq.~\eqref{eq main problem 1} in Lemma \ref{le: rewriting obj fun} below)
\begin{equation}
\max_{\varphi\in \Phi}\mathbb E\left[\int_0^T \left(\mu \pi_t-\frac{\gamma \sigma^2}{2}\pi_t^2\right)dt-\varepsilon\int_0^T\pi_t\frac{d\varphi^\downarrow_t}{\varphi_t} \right].
\end{equation}
From the outset, this objective is scale-invariant: doubling the initial number of risky shares and safe units, and also doubling the number of shares $\varphi_t$ held at time $t$ results in doubling also the number of safe units at time $t$ (through the self-financing condition), thereby leaving $d\varphi_t/\varphi_t$, $\pi_t = S_t \varphi_t / X_t$, and hence the objective, unchanged. 
Thus, we conjecture that the residual value function $V$ depends on the calendar time $t$ and on the variable $\zeta_t = \pi_t/(1-\pi_t)$, which denotes the number of shares held for each unit of the safe asset. In terms of this variable, the conditional value of the above objective at time $t$ becomes:
\begin{equation}\textstyle
\textstyle{F^\varphi(t) = \int_0^t \left(\mu \frac{\zeta_s}{1+\zeta_s} -\frac{\gamma \sigma^2}{2}\frac{\zeta_s^2}{(1+\zeta_s)^2}\right)ds-\varepsilon\int_0^t \frac{\zeta_s}{1+\zeta_s}\frac{d\varphi^\downarrow_s}{\varphi_s} 
+ V(t,\zeta_t)
.}
\end{equation}
By It\^o's formula, the dynamics of $F^\varphi$ is (henceforth the arguments of $V$ are omitted for brevity)
\begin{align*}
dF^\varphi(t) &= \left(\frac{\mu\zeta_t}{1+\zeta_t} -\frac{\gamma \sigma^2}{2}\frac{\zeta_t^2}{(1+\zeta_t)^2}\right)dt -\frac{\varepsilon\zeta_t}{1+\zeta_t}\frac{d\varphi^\downarrow_t}{\varphi_t} +V_tdt + V_\zeta d\zeta_t + \frac12 V_{\zeta\zeta}d\langle\zeta\rangle_t,
\end{align*}
where the subscripts of $V$ denote partial derivatives. Recall now the self-financing condition for the safe position $X_t$ and the risky position $Y_t$:
\begin{equation*}
dX_t = rX_tdt-S_td\varphi^\uparrow_t+(1-\varepsilon) S_t d\varphi^\downarrow_t,\quad dY_t = S_t d\varphi^\uparrow_t- S_td \varphi^\downarrow_t+\varphi_t dS_t
,
\end{equation*}
which implies the dynamics for the risky-safe ratio $\zeta_t$
\begin{equation*}
\frac{d\zeta_t}{\zeta_t} = \mu dt + \sigma dW_t + (1+\zeta_t) \frac{d\varphi_t}{\varphi_t} + \varepsilon \zeta_t \frac{d\varphi^\downarrow_t}{\varphi_t},
\end{equation*}
whence the dynamics of $F^\varphi$ simplifies to
\begin{align}
\textstyle{dF^\varphi(t) =}&\textstyle{ 
\left(\mu \frac{\zeta_t}{1+\zeta_t} -\frac{\gamma \sigma^2}{2}\frac{\zeta_t^2}{(1+\zeta_t)^2}
+V_t + \frac{\sigma^2}2 \zeta_t^2 V_{\zeta\zeta} + \mu \zeta_t V_\zeta
\right)dt}\\
-& \textstyle{\zeta_t \left(V_\zeta (1+(1-\varepsilon)\zeta_t) + \frac{\varepsilon}{1+\zeta_t} \right)\frac{d\varphi^\downarrow_t}{\varphi_t} 
+ \zeta_t (1+\zeta_t) V_\zeta \frac{d\varphi^\uparrow_t}{\varphi_t}
+\sigma \zeta_t V_\zeta dW_t.}
\end{align}
Now, by the martingale principle of optimal control \citep{davis1973dynamic} the process $F^\varphi(t)$ above needs to be a supermartingale for any trading policy $\varphi$ and a martingale for the optimal policy. As $\varphi^\uparrow$ and $\varphi^\downarrow$ are increasing processes, the supermartingale condition implies\footnote{In particular, the coefficients of  $\frac{d\varphi^\uparrow_t}{\varphi_t}$ and $\frac{d\varphi^\downarrow_t}{\varphi_t}$ need to be negative. As short positions are never optimal (cf. Remark \ref{lem integrability} and footnote \ref{foot special}), it follows that $0<\pi_t<1/\varepsilon$, whence only two cases arise: (a) $\zeta_t<-1/(1-\varepsilon)$, or (b) $\zeta_t>0$. In both cases $\zeta(1+\zeta)>0$ and $-\zeta (1+(1-\varepsilon)\zeta)<0$, whence \eqref{eq:boucon} follows.} the inequalities
\begin{align}\label{eq:boucon}
\textstyle{-\frac{\varepsilon}{(1+\zeta)(1+(1-\varepsilon)\zeta)} \le V_\zeta \le 0 ,}
\end{align}
and the martingale condition prescribes that the left (respectively, right) inequality becomes an equality at the points of increase of $\varphi^\downarrow$ (resp. $\varphi^\uparrow$). 
Likewise, it follows that 
\begin{equation*}
\textstyle{\mu\frac{\zeta}{1+\zeta} -\frac{\gamma \sigma^2}{2}\frac{\zeta^2}{(1+\zeta)^2}
+V_t + \frac{\sigma^2}2 \zeta^2 V_{\zeta\zeta} + \mu \zeta V_\zeta \le 0}
\end{equation*}
with the inequality holding as an equality whenever both inequalities in \eqref{eq:boucon} are strict. To achieve a stationary (that is, time-homogeneous) system, suppose that the residual value function  is of the form $V(t,\zeta) = \lambda (T-t) - \int^\zeta W(z)dz$ for some $\lambda$ to be determined, which represents the average optimal performance over a long period of time. Replacing this parametric form of the solution, the above inequalities become
\begin{align}
\textstyle{0\le  W(\zeta) \le \frac{\varepsilon}{(1+\zeta)(1+(1-\varepsilon)\zeta)},} \\
\textstyle{\mu \frac{\zeta}{1+\zeta} -\frac{\gamma \sigma^2}{2}\frac{\zeta^2}{(1+\zeta)^2}
-\lambda - \frac{\sigma^2}2 \zeta^2 W'(\zeta) - \mu \zeta W(\zeta) \le 0.}
\end{align}
Assuming further that the first inequality holds over some interval $[\zeta_-,\zeta_+]$, with each inequality reducing to an equality at the respective endpoint, the optimality conditions become
\begin{align}\label{eq:hjb}
\textstyle{\frac{\sigma^2}2 \zeta^2 W'(\zeta) + \mu \zeta W(\zeta) 
-\mu \frac{\zeta}{1+\zeta} +\frac{\gamma \sigma^2}{2}\frac{\zeta^2}{(1+\zeta)^2}
+\lambda  =}&\textstyle{ 0 \qquad \text{for }\zeta\in[\zeta_-,\zeta_+],}\\
\label{eq:bou1}
\textstyle{W(\zeta_-) = 0,\quad W(\zeta_+) = \frac{\varepsilon}{(\zeta_+ +1)(1+(1-\varepsilon)\zeta_+)}}&
,
\end{align}
which lead to a family of candidate value functions, each of them corresponding to a pair or boundaries ($\zeta_-, \zeta_+$). The optimal boundaries are identified by the smooth-pasting conditions, formally derived by differentiating eqs.~\eqref{eq:bou1} with respect to their boundaries
\begin{equation}
\label{eq:smooth1}
\textstyle{W'(\zeta_-) = 0,\quad W'(\zeta_+) = \frac{\varepsilon(\varepsilon-2(1-\varepsilon)\zeta_+-2)}{(1+\zeta_+)^2(1+(1-\varepsilon)\zeta_+)^2}.}
\end{equation}
These conditions identify the value function. The four unknowns are the free parameter in the general solution to the ordinary differential equation \eqref{eq:hjb}, the free boundaries $\zeta_-$ and $\zeta_+$, and the optimal rate $\lambda$. These quantities are identified by the boundary and smooth-pasting conditions \eqref{eq:bou1}--\eqref{eq:smooth1}.

\section{Conclusion}\label{sec: conclusion}

The costs of rebalancing a leveraged portfolio are substantial, and detract from its ostensible frictionless return. As leverage increases, such costs rise faster than the return, making it impossible for an investor to lever an asset's return beyond a certain multiple, net of trading costs. 

In contrast to the frictionless theory, trading costs make the risk-return trade-off nonlinear. An investor who seeks high return prefers an asset with high volatility to another one with equal Sharpe ratio but lower volatility, because higher volatility makes leverage cheaper to realize. A risk-neutral, return-maximizing investor does not take infinite leverage, but rather keeps it within a band that balances high exposure with low rebalancing costs.

\appendix

\section{Admissible Strategies}\label{sec: prereq}

In view of transaction costs, only finite-variation trading strategies are consistent with solvency. Denote by $X_t$ and $Y_t$ the wealth in the safe and risky positions respectively, and by $(\varphi_t^\uparrow)_{t\ge 0}$ and $(\varphi_t^\downarrow)_{t\ge 0}$ the cumulative number of shares bought and sold, respectively. The self-financing condition prescribes that $(X,Y)$ satisfy the dynamics
\begin{equation}\label{eq: sf1}
dX_t = rX_tdt-S_td\varphi^\uparrow_t+(1-\varepsilon) S_t d\varphi^\downarrow_t,\quad dY_t = S_t d\varphi^\uparrow_t- S_td \varphi^\downarrow_t+\varphi_t dS_t
.
\end{equation}

A strategy is admissible if it is nonanticipative and solvent, up to a small increase in the spread:
\begin{definition}\label{def: admiss}
Let $x>0$ (the initial capital) and let $(\varphi_t^\uparrow)_{t\ge 0}$ and $(\varphi_t^\downarrow)_{t\ge 0}$ be continuous, increasing processes, adapted to the augmented natural filtration of $B$. Then $(x, \;\varphi_t=\varphi_t^\uparrow-\varphi_t^\downarrow)$ is an \emph{admissible trading strategy} if 
\begin{enumerate}
\item its liquidation value is strictly positive at all times:
There exists $\varepsilon'>\varepsilon$ such that the discounted asset $\widetilde S_t:=e^{-rt}S_t$ satisfies
\begin{equation}\label{eq: liquidity}
x - \int_0^t \widetilde S_s d\varphi_s +\widetilde S_t \varphi_t - \varepsilon'\int_0^t \widetilde S_s d\varphi^\downarrow_s -\varepsilon'\varphi_t^+ \widetilde S_t > 0\qquad\text{a.s. for all }t\ge 0.
\end{equation}
 \item\label{part2 Lemma A21} 
The following integrability conditions hold\footnote{
Note that $\frac{\pi_t}{\varphi_t} = \frac{S_t}{w_t}$, therefore on the set $\{(\omega,t) : \varphi_t = 0\}$ the quantity $\frac{\pi_t}{\varphi_t}$ is well-defined.}
\begin{equation}\label{eq: admissible trading}
\mathbb E\left[\int_0^t \vert\pi_u\vert^2 du\right]<\infty,\quad\mathbb E\left[\int_0^t \pi_u \frac{d\|\varphi_u\|}{\varphi_u}\right]<\infty \quad\text{ for all }t\ge 0,
\end{equation}
where $\|\varphi_t\|$ denotes the total variation of $\varphi$ on $[0,t]$. 
\end{enumerate}
The family of admissible trading strategies is denoted by $\Phi$.
\end{definition}

The following lemma describes the dynamics of the wealth process $w_t$, the risky weight $\pi_t$, and the risky-safe ratio $\zeta_t$. 
\begin{lemma}\label{le: rewriting obj fun}
For any admissible trading strategy $\varphi$, \footnote{The notation $\frac{dx_t}{x_t}=dy_t$ means $x_t=x_0+\int_0^t x_s dy_s$, hence the SDEs are well defined even for null $x_t$.}
\begin{align}\label{eq zeta diff}
\frac{d\zeta_t}{\zeta_t}&=\mu dt +\sigma dB_t+(1+\zeta_t)\frac{d\varphi_t^\uparrow}{\varphi_t}    -(1+(1-\varepsilon)\zeta_t)\frac{d\varphi_t^\downarrow}{\varphi_t},\\\label{eq w diff}
\frac{dw_t}{w_t}&=r dt+\pi_t(\mu dt+\sigma dB_t-\varepsilon \frac{d\varphi_t^\downarrow}{\varphi_t}),\\\label{eq pi diff}
\frac{d\pi_t}{\pi_t}&=(1-\pi_t)(\mu dt+\sigma dB_t)-\pi_t(1-\pi_t)\sigma^2 dt+\frac{d\varphi_t^\uparrow}{\varphi_t}-(1-\varepsilon\pi_t)\frac{d\varphi_t^\downarrow}{\varphi_t}.
\end{align}
For any such strategy, the functional
\begin{equation}\label{eq: finite mean var fun}
F_T(\varphi):=
\frac1 T \mathbb E\left[ \int_0^T \frac{dw_t}{w_t}-\frac\gamma 2\left\langle  \int_0^T \frac{d w_t}{w_t}\right\rangle_T\right]
\end{equation}
equals to
\begin{equation}\label{eq main problem 1}
F_T(\varphi)=r+\frac{1}{T}\mathbb E\left[\int_0^T \left(\mu \pi_t-\frac{\gamma \sigma^2}{2}\pi_t^2\right)dt-\varepsilon\int_0^T\pi_t\frac{d\varphi^\downarrow_t}{\varphi_t} \right].
\end{equation}
\end{lemma}
\begin{proof}
The self-financing conditions \eqref{eq: sf1} imply that
\begin{align}\label{eq: cash we have}
\frac{dX_t}{X_t} =& rdt-\zeta_t \frac{d\varphi_t^\uparrow}{\varphi_t}+(1-\varepsilon) \zeta_t \frac{d\varphi_t^\downarrow}{\varphi_t},
\\
\label{eq: stock we have}
\frac{dY_t}{Y_t} =& \frac{d\varphi_t^\uparrow}{\varphi_t}-\frac{d\varphi_t^\downarrow}{\varphi_t}+\frac{dS_t}{S_t},
\\
\label{eq: Ito app}
\frac{d(Y_t/X_t)}{Y_t/X_t} =&
\frac{dY_t}{Y_t}-\frac{dX_t}{X_t}+\frac{d\langle X\rangle_t}{X^2_t}-
\frac{d\langle X,Y\rangle_t}{X_t Y_t}=\frac{dY_t}{Y_t}-\frac{dX_t}{X_t}.
\end{align}
Equation \eqref{eq zeta diff} follows from the last equation, and \eqref{eq  w diff} holds in view of equations \eqref{eq: cash we have} and \eqref{eq: stock we have}. Equation \eqref{eq pi diff} follows from the identity $\pi_t=1-\frac{1}{1+\zeta_t}$ and \eqref{eq zeta diff}.
The expression in \eqref{eq main problem 1} for the objective functional follows from equation \eqref{eq w diff}.
\end{proof}

The following lemma shows that, without loss of generality, it is enough to consider trading strategies which do not take short positions in the risky asset. 
\begin{lemma}\label{lem integrability}
If $\varphi\in\Phi$ is optimal for \eqref{eq: obj asympt}, then also the strategy $\hat{\varphi}_t:=\varphi_t 1_{\{\varphi_t\geq 0\}}$  is optimal.
\end{lemma}

\begin{proof}
Due to Lemma \ref{le: rewriting obj fun}, the objective functional has the equivalent form \eqref{eq main problem 1},  (letting $T\rightarrow\infty$).  It is clear that $\hat\varphi$ is an admissible trading strategy if $\varphi$ is. Furthermore, as $\mu\geq 0$,
$\mu\hat\pi_t\geq \mu\pi_t$ at all times $t$, and $\hat \pi_t=0$ whenever $\varphi_t<0$,
 whence $F_T(\hat\varphi)\geq F_T(\varphi)$ for each $T>0$.
\end{proof} 

\begin{remark}\label{rem: range of trade}
In view of this Lemma and admissibility, it suffices to consider trading strategies which satisfy $0\leq \pi_t\leq 1/\varepsilon$, or, in terms of the risky-safe ratio, $\zeta_t<-1/(1-\varepsilon)$ or $\zeta_t\geq 0$.
\end{remark}

\section{Risk Aversion and Efficient Frontier}\label{appendix HJB}
This section contains a series of propositions that lead to the proof of Theorem \ref{thm free b} \ref{thm free b part 1}--\ref{thm free b part 3}. Part \ref{thm free b part 4} of the theorem is postponed to Appendix \ref{app proof main theorem}.
Set
\begin{equation}
G(\zeta):=\frac{\varepsilon}{(1+\zeta)(1+(1-\varepsilon)\zeta)}
\qquad\text{and}\qquad
h(\zeta):=\mu \left(\frac{\zeta}{1+\zeta}\right)-\frac{\gamma \sigma^2}{2}\left(\frac{\zeta}{1+\zeta}\right)^2.
\end{equation}
Defining $H:=h'$,  the free boundary problem \eqref{eq: TKA fbp}--\eqref{terminal1 TKA}
reduces to
\begin{align}\label{eq: TKA fbp re}
&\frac{1}{2}\sigma^2 \zeta^2 W''(\zeta)+(\sigma^2+\mu)\zeta W'(\zeta)+\mu W(\zeta)-H(\zeta)=0,\\\label{initial0 TKA re}
&W(\zeta_-)=0,\\\label{initial1 TKA re}
&W'(\zeta_-)=0,\\\label{terminal0 TKA re}
&W(\zeta_+)=G(\zeta_+),\\\label{terminal1 TKA re}
&W'(\zeta_+)=G'(\zeta_+).
\end{align}
\begin{proposition}\label{prop1}
Let $\gamma>0$ and $\pi_*\neq 1$. For sufficiently small $\varepsilon$, the free boundary problem \eqref{eq: TKA fbp re}--\eqref{terminal1 TKA re}  has a unique solution $(W,\zeta_-,\zeta_+)$, with $\zeta_-<\zeta_+$.  The free boundaries have the asymptotic expansion
\begin{equation}\label{eq: asymptotics zeta}
\textstyle{\zeta_\pm=\frac{\pi_*}{1-\pi_*}\pm\left(\frac{3}{4\gamma}\right)^{1/3}\left(\frac{\pi_*}{(\pi_*-1)^2}\right)^{2/3} \varepsilon^{1/3}-\frac{(5-2\gamma)\pi_*}{2\gamma(\pi_*-1)^2}\left(\frac{\gamma \pi_*(\pi_*-1)}{6}\right)^{1/3}\varepsilon^{2/3}+O(\varepsilon).}
\end{equation}
\end{proposition}

\begin{proof}[Proof of Proposition \ref{prop1}]
Note that \eqref{eq: TKA fbp re} is equivalent to the ODE
\[
\left(\frac{\sigma^2 \zeta^2}{2}W'(\zeta)+\mu \zeta W(\zeta)-h(\zeta)\right)'=0
\]
and thus, the initial conditions \eqref{initial0 TKA re}, \eqref{initial1 TKA re} imply
that $W$ satisfies 
\[
\frac{\sigma^2 \zeta^2}{2}W'(\zeta)+\mu \zeta W(\zeta)=h(\zeta)-h(\zeta_-), \quad W(\zeta_-)=0.
\]
By the variation of constants method, and as $\zeta_-\notin\{-1,0\}$, any solution of the
initial value problem \eqref{eq: TKA fbp re}--\eqref{initial1 TKA re} is thus of the form
\begin{equation}\label{TKA IVP sol}
\widetilde W(\zeta_-,\zeta):=\frac{2}{(\sigma\zeta)^2 }\int_{\zeta_-}^\zeta(h(y)-h(\zeta_-))\left(\frac{y}{\zeta}\right)^{2\gamma \pi_*-2}dy.
\end{equation}
Suppose $(W,\zeta_-,\zeta_+)$ is a solution of \eqref{eq: TKA fbp re}--\eqref{terminal1 TKA re}. In view of \eqref{TKA IVP sol},
$ W(\cdot)\equiv \widetilde W(\zeta_-,\cdot)$. Let
\begin{equation}\label{ Integral J}
 J(\zeta_-,\zeta):=\frac{\sigma^2 \zeta^{2\gamma \pi_*}}{2}\widetilde W(\zeta_-,\zeta).
\end{equation}
By the terminal conditions \eqref{terminal0 TKA re}--\eqref{terminal1 TKA re} at $\zeta_+$, and setting $\delta=\varepsilon^{1/3}$, 
 $(\zeta_-,\zeta_+)$ satisfy the following system of algebraic equations,
\begin{align}\label{eq: TAK1}
\Psi_1(\zeta_-,\zeta_+):=& \widetilde W(\zeta_-,\zeta_+)-\frac{\delta^3}{(1+\zeta_+)(1+(1-\delta^3)\zeta_+)}=0,\\
\label{eq: TAK2}
\Psi_2(\zeta_-,\zeta_+):=& \textstyle
\frac{2 (h(\zeta_+)-h(\zeta_-))}{\sigma^2 \zeta_+^2}-\frac{2\gamma \pi_*}{\zeta_+} \widetilde W(\zeta_-,\zeta_+)-\frac{(1-\delta^3)^2}{(1+(1-\delta^3)\zeta_+))^2}+\frac{1}{(1+\zeta_+)^2}=0.
\end{align}
Conversely, if $(\zeta_-,\zeta_+)$ solve \eqref{eq: TAK1}--\eqref{eq: TAK2}, then the triplet
$(\zeta\mapsto \widetilde W(\zeta_-,\zeta),\zeta_-,\zeta_+)$ provides a solution to the free boundary problem \eqref{eq: TKA fbp re}--\eqref{terminal1 TKA re}. Therefore,
to provide a unique solution of the free boundary problem, it suffices to provide a unique solution of  \eqref{eq: TAK1}--\eqref{eq: TAK2}.

To obtain a guess for the asymptotic expansions of the prospective solutions $\zeta_\pm$, expand $\Psi_{1,2}$ around
\begin{equation}\label{guezz}
\zeta_{-}=\zeta_*+B_{1}\delta+O(\delta^2), \quad \quad \zeta_+=\zeta_*+B_2\delta+O(\delta^2),\quad\text{where}\quad\zeta_*=\frac{\pi_*}{1-\pi_*},
\end{equation}
which yields 
\begin{align}\label{eq: mr no}
\Psi_1(\zeta_\pm(\delta))&=-\frac{\gamma(1-\pi_*)^6}{3\pi_*^2}\left(2 B_1^3-3 B_1^2 B_2+B_2^3+\frac{3\pi_*^2}{\gamma (1-\pi_*)^4}\right)\delta^3+O(\delta^4),\\\label{psi2 estimate}
\Psi_2(\zeta_\pm(\delta))&=\frac{(B_1-B_2)(B_1+B_2)\gamma(\pi_*-1)^6}{\pi_*^2}\delta^2+O(\delta^3).
\end{align}
Equating the coefficients of the leading order terms to zero yields\footnote{The coefficient in \eqref{psi2 estimate} vanishes also for $B_1=B_2$, but \eqref{eq: mr no} does not, excluding such a case.} 
\begin{align}\label{eq: no1}
2 B_1^3-3 B_1^2 B_2+B_2^3+\frac{3\pi_*^2}{\gamma (1-\pi_*)^4}&=0,\\\label{eq: no2}
B_1+B_2&=0,
\end{align}
whence $B_1=-B_2$ and solves $B_1^3=-\frac{3}{4\gamma}\frac{\pi_*^2}{(1-\pi_*)^4}=0$, and thus
\begin{equation}\label{eq B1}
B_1=-\left(\frac{3}{4\gamma}\right)^{1/3}\left(\frac{\pi_*}{(1-\pi_*)^2}\right)^{2/3}.
\end{equation}
With the change of variables
\begin{equation}\label{def of eta}
\eta_\pm:=\frac{\zeta_\pm-\zeta_*}{\delta}
\end{equation}
and the notation
\begin{equation}\label{eq: system psi12def}
\Phi_{1}(\eta_-,\eta_+):=\Psi_{1}(\zeta_-(\eta_-),\zeta_+(\eta_+)),\quad \Phi_{2}(\eta_-,\eta_+):=\Psi_{2}(\zeta_-(\eta_-),\zeta_+(\eta_+))
\end{equation}
the system \eqref{eq: TAK1}--\eqref{eq: TAK2} for $\zeta_\pm$ reduces to 
\begin{equation}\label{ana sys}
\Phi(\eta_-,\eta_+)=(\Phi_1(\eta_-,\eta_+),\Phi_2(\eta_-,\eta_+))=0
\end{equation}
in the unknowns $\eta_\pm$. Because of \eqref{eq B1}, the guess \eqref{guezz} takes the
explicit form
\begin{equation}\label{guezzx}
\zeta_{\pm}=\zeta_*\pm\left(\frac{3}{4\gamma}\right)^{1/3}\left(\frac{\pi_*}{(1-\pi_*)^2}\right)^{2/3}\delta+O(\delta^2), 
\end{equation}
which suggests that the solution $(\eta_-,\eta_+)$ is around $(B_1,B_2=-B_1)$. Proposition \ref{prop eqpsi12} below indeed guarantees the existence a unique solution
around  around $(B_1,B_2=-B_1)$ for sufficiently small $\delta>0$, which is analytic in $\delta$. Hence, also the original system
$\Psi(\zeta_-,\zeta_+)=0$ has a unique solution $(\zeta_-,\zeta_+)$ for small $\delta$, with the first
order proxies \eqref{guezzx}. This implies that the free boundary problem \eqref{eq: TKA fbp re}--\eqref{terminal1 TKA re} has a unique solution for sufficiently small $\varepsilon$.

%
%
%
%

To derive the higher order terms of \eqref{eq: asymptotics zeta}, it is useful to rewrite the integral \eqref{ Integral J} as\footnote{For  $\pi_* = 1/(2\gamma)$, $I_1 = h(\zeta_-)(\log \zeta_--\log \zeta_+)$ and $I_2 = \int_{\zeta_-}^{\zeta_+} \frac{h(y)}{y} dy$.}
\begin{equation}\label{life aid1}
J(\zeta_-,\zeta_+)=
\underbrace{\frac{h(\zeta_-)(\zeta_-^{2\gamma \pi_*-1}-\zeta_+^{2\gamma \pi_*-1})}{2\gamma \pi_*-1}}_{=:I_1}
+
\underbrace{\int_{\zeta_-}^{\zeta_+} h(y)y^{2\gamma \pi_*-2}dy}_{=:I_2}.
\end{equation}
The derivative of $I_2$ with respect to $\delta$ equals
\begin{equation}\label{life aid2}
\frac{dI_2}{d\delta}=h(\zeta_+) \zeta_+^{2\gamma \pi_*-2}\frac{d \zeta_+}{d\delta}-h(\zeta_-)\zeta_-^{2\gamma \pi_*-2}\frac{d\zeta_-}{d\delta}.
\end{equation}
Now, expanding the right-hand side as a power series in $\delta$, and integrating with respect to $\delta$ yields an asymptotic expansion of $I_2$. 

To obtain these expansions, guess a solution of equations \eqref{eq: TAK1}--\eqref{eq: TAK2} of the form
\[
\zeta_\pm=\frac{\pi_*}{1-\pi_*}\pm\left(\frac{3}{4\gamma}\right)^{1/3}\left(\frac{\pi_*}{(1-\pi_*)^2}\right)^{2/3}\delta+A_{\pm}\delta^2+O(\delta^3),\\
\]
for some unkowns $A_\pm$, and substitute it into equations \eqref{eq: TAK1}--\eqref{eq: TAK2}, thereby using \eqref{life aid1} and \eqref{life aid2}. Comparing the coefficients in the asymptotic expansion of the two equations reveals that
\[
A_-=A_+=\left(\frac{(5-2\gamma)\pi_*}{2\gamma(1-\pi_*)^2}\right)\left(\frac{\gamma \pi_*(1-\pi_*)}{6}\right)^{1/3},
\]
and therefore \eqref{eq: asymptotics zeta} holds.
\end{proof}

\begin{proposition}\label{prop eqpsi12}
Let $\gamma>0$ and $\pi_*\neq 1$, and recall $B_1=-B_2$ from \eqref{eq B1}. For sufficiently small $\delta>0$, the system \eqref{ana sys}, where $\Phi=(\Phi_1,\Phi_2)$ is defined by \eqref{eq: system psi12def}, has a unique solution $(\eta_-(\delta),\eta_+(\delta))$
satisfying $\eta_-(0)=B_1$, $\eta_+(0)=B_2$, and $\delta\mapsto \eta_\pm(\delta)$ are analytic functions.
\end{proposition}
\begin{proof}
Consider first the ``general'' case $\mu/\sigma^2\neq 1/2$: Introduce the rescaled functions $\widetilde{\Phi}_{1,2}$ and $\widetilde{\Phi}:=(\widetilde \Phi_1,\widetilde \Phi_2)$ defined as
\begin{equation}\label{eq: system psi12}
\widetilde\Phi_{1}:=\frac{\Phi_{1}}{\delta^l},\quad \widetilde\Phi_{2}:=\frac{\Phi_{2}}{\delta^m},
\end{equation}
where $ l=3$ and $m=2$. By scaling, the function $\widetilde \Phi$ depends on three arguments, and for the sake of clarity henceforth it is denoted by 
\[
\widetilde \Phi=\widetilde\Phi(\eta_-,\eta_+,\delta)
\]
Let $D\widetilde\Phi$ be the Frechet differential of $\widetilde\Phi$. As shown next, the Jacobian satisfies, 
\begin{equation}\label{eq: jacobi cookies}
\det (D\widetilde\Phi)(\eta_-=B_1,\eta_+=B_2,\delta=0)=\frac{6 \gamma (1-\pi_*)^8(2\gamma \pi_*-1)}{\pi_*^2}\neq 0,
\end{equation}
hence the implicit function theorem for analytic functions \cite[Theorem I.B.4]{GunningRossi} ensures that for sufficiently small $\delta$ 
there exists a unique solution $(\eta_-,\eta_+)$ of $\widetilde\Phi(\eta_-,\eta_+)=0$ around $(B_1,B_2)$ which is analytic in $\delta$. 

It remains to prove \eqref{eq: jacobi cookies}. By construction,
\[
\Psi_2(\zeta_-,\zeta_+)=\frac{\partial \Psi_1(\zeta_-,\zeta_+)}{\partial \zeta_+},
\]
whence
\begin{align*}
\frac{\partial \widetilde\Phi_1}{\partial \eta_+}(\eta_\pm)&=\frac{1}{\delta^l}\frac{\partial \Phi_1}{\partial \eta_+}(\eta_\pm)=\frac{1}{\delta^l}\frac{\partial \Psi_1(\zeta_\pm(\eta_\pm))}{\partial \zeta_+}\frac{\partial \zeta_+}{\partial \eta_+}\\&=\frac{1}{\delta^{l-1}}\frac{\partial \Psi_1(\zeta_\pm(\eta_\pm))}{\partial \zeta_+}=\frac{\Psi_2 (\zeta_\pm(\eta_\pm))}{\delta^{l-1}}
\end{align*}
and thus inserting the definition of $\eta_\pm$ (cf.~\eqref{def of eta}) into equation \eqref{eq: TAK2} and, letting $\delta\rightarrow 0$, in view of \eqref{eq: no1}--\eqref{eq: no2} and their solution \eqref{eq B1} it follows that 
\[
\frac{\partial \widetilde\Phi_1}{\partial \eta_+}\mid_{(B_1,B_2,0)}=0.
\]
Thus the determinant of the Jacobian is simply
\[
\det(D\widetilde\Phi)(B_1,B_2,0)=\frac{\partial \widetilde \Phi_1(\eta_-,\eta_+)}{\partial \eta_-}\mid_{(B_1,B_2,0)}\times\frac{\partial \widetilde \Phi_2(\eta_-,\eta_+)}{\partial \eta_+}\mid_{(B_1,B_2,0)}.
\]
Because
\[
\frac{\partial \Psi_1}{\partial \zeta_-}=-\frac{2h'(\zeta_-)}{\sigma^2 \zeta_+^{2\mu/\sigma^2}}\left(\frac{\zeta_+^{2\mu/\sigma^2-2}}{2 \mu/\sigma^2-1}-\frac{\zeta_-^{2\mu/\sigma^2-2}}{2 \mu/\sigma^2-1}\right)
\]
and by the chain rule
\[
\frac{\partial \widetilde \Phi_1(\eta_-,\eta_+)}{\partial \eta_-}=\frac{1}{\delta_3}\frac{\partial \Psi_1}{\partial \zeta_-} \times \delta,
\]
it follows that

\[
\frac{\partial \widetilde\Phi_1(\eta_-,\eta_+)}{\partial \eta_-}\mid_{(B_1,B_2,0)}=\frac{6^{2/3}(1-\pi_*)^3(\gamma \pi_*(1-\pi_*))^{1/3}(1-2\gamma \pi_*)}{\pi_*}.
\]
Similarly, 
\[
\frac{\partial \widetilde\Phi_2(\eta_-,\eta_+)}{\partial \eta_+}\mid_{(B_1,B_2,0)}=-\frac{6^{1/3}(1-\pi_*)^4(\gamma (1-\pi_*)\pi_*)^{2/3}}{\pi_*^2}, 
\]
from which \eqref{eq: jacobi cookies} and hence the assertion in the proposition follows.

For the ``singular'' case $\mu/\sigma^2=1/2$ one needs to set $l=5$, $m=3$ in \eqref{eq: system psi12}, then the right side of \eqref{eq: jacobi cookies} equals $3 (1-1/\pi_*)^8 \pi_*^5\neq 0$, and therefore similar arguments as in the general case apply.
\end{proof}

\begin{definition}\label{definition HJB eq}
A solution of the HJB equation is a pair $(V,\lambda)$, where $V$ is a twice continuously differentiable function, which satisfies
\begin{equation}\label{eq: HJBx}
\min(\mathcal A V(x)-h(x)+\lambda, G(x)-V'(x),V'(x))=0, \quad x\in\left (-\infty, -\frac{1}{1-\varepsilon}\right)\cup (0,\infty),
\end{equation}
where $\mathcal A : \mathcal C^2(\mathbb R) \mapsto \mathcal C^2(\mathbb R)$ is the differential operator
\[
\mathcal A f(x):=\frac{\sigma^2}{2}x^2 f''(x)+\mu x f'(x).
\] 
\end{definition}

Note that the restriction $x\in\left (-\infty, -\frac{1}{1-\varepsilon}\right)\cup (0,\infty)$ is motivated by Remark \ref{rem: range of trade}.

\begin{proposition}\label{prop2}
Let $(W,\zeta_-,\zeta_+)$ be the solution of the free boundary problem \eqref{terminal0 TKA re}--\eqref{terminal1 TKA re}  (provided by Proposition \ref{prop1}) with asymptotic expansion 
\eqref{eq: asymptotics zeta}. For sufficiently small $\varepsilon$, the pair 
\[
V(\cdot):=\int_0^{\cdot} \hat W(\zeta)d\zeta, \quad \lambda:=h(\zeta_-),
\]
where
\begin{equation}\label{eq candidate HJB}
\hat W(\zeta):=\begin{cases} 0 \quad &\text{for}\quad\zeta<\zeta_-,\\
W(\zeta) \quad &\text{for}\quad\zeta \in [\zeta_-,\zeta_+],\\
G(\zeta)\quad &\text{for}\quad \zeta\geq \zeta_+,\end{cases}
\end{equation}
is a solution of the HJB equation \eqref{eq: HJBx}.
\end{proposition}

\begin{proof}[Proof of Proposition \ref{prop2}]
To check that $(V,\lambda)$ solves the HJB equation \eqref{eq: HJBx}, consider separately the domains $[\zeta_-,\zeta_+]$, $\zeta<\zeta_-$ and $\zeta>\zeta_+$. From the decompositions
\[
G(\zeta)=\frac{1}{1+\zeta}-\frac{1-\varepsilon}{1+(1-\varepsilon)\zeta}
\qquad\text{and}\qquad
G'(\zeta)=\left(\frac{1-\varepsilon}{1+(1-\varepsilon)\zeta}\right)^2-\frac{1}{(1+\zeta)^2}
,
\]
note first that on $[\zeta_-,\zeta_+]$, by construction it holds that
\[
(\mathcal A V(\zeta)-h(\zeta)+h(\zeta_-))'=\frac{1}{2}\sigma^2 \zeta^2 W''(\zeta)+(\sigma^2+\mu)\zeta W'(\zeta)+\mu W(\zeta)-H(\zeta)=0.
\]
Furthermore, in view of the initial conditions \eqref{initial0 TKA re}--\eqref{initial1 TKA re},
\[
(\mathcal A V(\zeta)-h(\zeta)+h(\zeta_-))\mid_{\zeta=\zeta_-}=\mathcal A V(\zeta)\mid_{\zeta=\zeta_-}=0,
\]
whence 
\[
\mathcal A V(\zeta)-h(\zeta)+h(\zeta_-)\equiv 0,\quad \zeta\in [\zeta_-,\zeta_+].
\]
To see that $0\leq V'\leq G$ on all of $[\zeta_-,\zeta_+]$, observe that
\begin{equation}\label{der h}
(h(\zeta)-h(\zeta_-))'=h'(\zeta)=H(\zeta)=\frac{\mu}{\pi_*(1+\zeta)^2}\left(\pi_*-\frac{\zeta}{1+\zeta}\right).
\end{equation}

Note that for $\zeta_-< \zeta\leq \zeta^*$, where $\zeta^*/(1+\zeta^*)=\pi_*$,  
$V'(\zeta)=W(\zeta)>0$. It is shown that also $W(\cdot)\geq 0$ on all of $[\zeta_-,\zeta_+]$. This is equivalent to showing non-negativity
of
\begin{equation}\label{simple W}
w(\zeta):=2 \sigma^2 \zeta^{2\gamma\pi_*} W(\zeta)=\int_{\zeta_-}^\zeta (h(x)-h(\zeta_-))x^{2\gamma\pi_*-2}dx.
\end{equation}
Now $w'(\zeta)= (h(\zeta)-h(\zeta_-))\zeta^{2\gamma\pi_*-2}=0$ if and only if $h(\zeta_-)=h(\zeta)$. Hence, either $\zeta=\zeta_-$ or $\zeta=\overline\zeta$, where
\[
\pi(\overline\zeta)=\frac{\overline\zeta}{1+\overline\zeta}=2\pi_*-\pi_-.
\]
By the first-order asymptotics of \eqref{eq: asymptotics zeta}, one obtains $\bar \zeta\notin [\zeta_-,\zeta_+]$ for sufficiently small $\varepsilon$. Therefore $w'>0$ on $(\zeta_-,\zeta_+]$, and by \eqref{simple W} it follows that $V'\geq 0$ on all of $[\zeta_-,\zeta_+]$. 
To conclude the validity of the HJB equation on $[\zeta_-,\zeta_+]$, it only remains to show the inequality $V'\leq G$. To this end, notice that $\Psi_1(\zeta)=W(\zeta)-G(\zeta)$, (this is the function defined in \eqref{eq: TAK1}, with fixed $\zeta_-$) satisfies
\[
\Psi_1(\zeta_-)=-G(\zeta_-)=-\frac{\varepsilon}{(1+\zeta_-)(1+(1-\varepsilon)\zeta_-)}=-(1-\pi_*)^2\varepsilon+O(\varepsilon^{4/3}),
\]
hence for sufficiently small $\varepsilon$, $\Psi_1(\zeta)<0$ on some interval
$[\zeta_-, \bar \zeta)$, and $\Psi_1(\bar \zeta)=0$. Therefore, $\bar\zeta\leq \zeta_+$. As
$\Psi_1( \zeta_+)=0$ by construction, it suffices to show that $\bar \zeta=\zeta_+$ to prove non-negativity
of $\Psi_1$ on $[\zeta_-,\zeta_+]$. Suppose, by contradiction, that there exists a sequence $\delta_k\downarrow 0$ such that for each $k\geq 1$, $\Psi_1(\bar\zeta(\delta_k))=0$, and that $\zeta_-(\delta_k)<\bar \zeta(\delta_k) <\zeta_+(\delta_k)$. Now, change variable to $u=\frac{\zeta-\zeta_*}{\delta}$, and introduce the notation $u_\pm=\frac{\zeta_\pm-\zeta_*}{\delta}$, $\bar u=\frac{\bar\zeta-\zeta_*}{\delta}$. Up to a subsequence, without loss of generality assume that $\bar u(\delta_k)$ converges, whence it satisfies
\[
\lim_{k\rightarrow\infty}\bar u(\delta_k)=:B_0\in [B_1,B_2],
\]
where $B_1$ is defined in \eqref{eq B1}, and $B_2=-B_1$. The calculations leading to \eqref{eq B1} therefore entail that $B_0$ must satisfy
\eqref{eq: no1} in place of $B_2$, i.e.
\begin{equation}\label{eq superduper}
2 B_1^3-3 B_1^2 B_0+B_0^3+\frac{3\pi_*^2}{\gamma (1-\pi_*)^4}=0.
\end{equation}
With $B_1$ from \eqref{eq B1} and the change of variable
$\xi=-B_0/B_1$ implies $2-3\xi+\xi^3=0$ which has the only solutions $1$ and $-2$. Therefore, \eqref{eq superduper} has the only relevant solution
\[
B_0=-B_1=B_2.
\]
By intertwining $u_+(\delta)$ and $\bar u(\delta_k)$, one can introduce 
\[
\bar u^*(\delta)=\begin{cases}\bar u(\delta_k),\quad k\in\mathbb N\\ u_+(\delta),\quad \text{otherwise}\end{cases}.
\]
Hence $(u_-(\delta),u^*(\delta))$ satisfies $\Phi(u_-,u_+)=0$ near $(B_1,B_2)$, for sufficiently small $\delta$. By Proposition \ref{prop eqpsi12},
$u^*(\delta)=u_+(\delta)$, which contradicts our assumption $\bar\zeta\neq \zeta_+$.

Consider now $\zeta\leq \zeta_-$. $V$ solves the HJB equation, if 
\[
\mathcal A V-h(\zeta)+h(\zeta_-)=h(\zeta_-)-h(\zeta)\geq 0, \quad G(\zeta)\geq 0.
\]
As $h(\zeta)-h(\zeta_-)=0$ for $\zeta=\zeta_-$,  it suffices to show that $h'$ is non-negative to obtain the first inequality. To this end, the explicit
formula  \eqref{der h} for the derivative is used. Now for small $\varepsilon$
clearly $\pi_-<\pi_*$, hence for $\zeta=\zeta_-$ \eqref{der h} is indeed strictly positive, hence, upon integration, one obtains the first inequality for any $\zeta<\zeta_-$. To settle the second inequality, recall that either $\zeta<-1/(1-\varepsilon)$ or $\zeta>0$. On these domains, $G$ is clearly a strictly positive function.  Hence it is proved that $V$ satisfies the HJB equation
for $\zeta\leq \zeta_-$.

Finally, consider $\zeta\geq \zeta_+$. As $G=W$, it suffices to show
\begin{equation}\label{eq part 3 first eq}
L(\zeta):=\mathcal A V(\zeta)-h(\zeta)+h(\zeta_-)\geq 0,\quad G(\zeta)\geq 0.
\end{equation}
As $G(\zeta)$ is strictly positive, the second inequality holds. For the first inequality in \eqref{eq part 3 first eq}, note that
\begin{align*}
L(\zeta)&=\frac{\sigma^2 \zeta^2}{2}G'(\zeta)+\mu \zeta G(\zeta)-h(\zeta)+h(\zeta_-)
\end{align*}
and $L(\zeta_+)=0$, because of \eqref{eq: TKA fbp re}, \eqref{terminal0 TKA re} and \eqref{terminal1 TKA re}. Therefore it suffices to show $L$ has no zeros on $[\zeta_+,-1/(1-\varepsilon))$, besides $\zeta_+$. 

Consider first, $\gamma=1$. Using the transformation $z=\frac{\zeta}{1+\zeta}$ one can rewrite $L$ in terms of $z$, denoting it by 
$F(z,\varepsilon):=L(\zeta(z))$. As $F(\pi_+)=0$, polynomial division by $(z-\pi_+)$ yields
\begin{equation}\label{eq: Dec F}
F(z,\varepsilon)=\frac{(z-\pi_+)}{(1-\varepsilon z)^2}g(z),
\end{equation}
and $g(z)=\frac{1}{2}(g_0+g_1 z)$, where
\begin{align*}
g_0&=2\mu (-1+(1-2\pi_-+\pi_+)\varepsilon-(1-\pi_-)\pi_+\varepsilon^2\\&\qquad+\sigma^2(\pi_++2(\pi_-^2-\pi_+)\varepsilon+\pi_+(1-\pi_-^2)\varepsilon^2),\\
g_1&=1-(1-\pi_-)\varepsilon)(\sigma^2+\varepsilon(2\mu-(1+\pi_-)\sigma^2).
\end{align*}
Therefore, the following asymptotic expansions hold
\[
g(\pi_+)=\sigma^2\left(\frac{3}{4}\left(\frac{\mu}{\sigma^2}\right)^2\left(1-\frac{\mu}{\sigma^2}\right)^2\right)^{1/3}\varepsilon^{1/3}+O(\varepsilon^{2/3}), \quad g(1/\varepsilon)=\frac{\sigma^2}{2\varepsilon}+O(1).
\]
It follows that $g$ has no zeros on $[\pi_+,1/\varepsilon]$, for sufficiently small $\varepsilon$. Hence 
$F(z)>0$ for $z\in (\pi_+,1/\varepsilon)$.

Next, consider $\gamma\neq 1$. Using the transformation $z=\frac{\zeta}{1+\zeta}$ one can rewrite, similar to the case $\gamma=1$, $L$ in terms of $z$, obtaining the function $F(z,\varepsilon)=L(\zeta(z))$. It is proved next that $F$ has no zeros on $(\pi_+,1/\varepsilon)$.

As $F(\pi_+)=0$, polynomial division by $(z-\pi_+)$ yields \eqref{eq: Dec F},
where the third order polynomial $g$ has derivative
\[
g'=a_0+a_1z+a_2 z^2,
\]
where the coefficients $a_0,a_1$ and $a_2$ are complex, yet explicit, functions of the parameters and the relative bid-ask spread $\varepsilon$.

In view of \eqref{eq: Dec F}, it is enough to show that $g$ has no zeros on $[\pi_+,1/\varepsilon]$. First, note the following asymptotic expansions,
\begin{equation}\label{g at pip}
g(\pi_+)=\left(\frac{3}{4\gamma}\pi_*^2(\pi_*-1)^2\right)^{1/3}\varepsilon^{1/3}+O(\varepsilon^{2/3}),\quad
g(1/\varepsilon)=\frac{\sigma^2}{2\varepsilon}+O(1).
\end{equation}
Therefore, for sufficiently small $\varepsilon$, $g>0$ on both endpoints of $[\pi_+,1/\varepsilon]$. It remains to show
that any local minimum of $g$ in  $[\pi_+,1/\varepsilon]$ is non-negative. The local extrema $z_\pm$, where $g'(z_\pm)=0$, have asymptotic expansions $z_\pm=\frac{2}{3\varepsilon}\pm \frac{1}{3\varepsilon}\sqrt{\frac{\gamma-4}{\gamma-1}}+O(1).$
Obviously, there are no local extrema in $[\pi_+,1/\varepsilon]$ whenever $\gamma \in[1,4)$. Therefore $g>0$
on all of $[\pi_+,1/\varepsilon]$, and thus $F(z)\geq 0$ on $[\pi_+,1/\varepsilon)$. The non-trivial case $\gamma \notin[1,4)$ remains:

For $0<\gamma<1$ it holds that $\frac{4-\gamma}{1-\gamma}>4$, hence $z_\pm\notin [\pi_+,1/\varepsilon]$. It follows that $g'$
has no zeros in this interval and thus $g>0$ on  $[\pi_+,1/\varepsilon]$.

Next, consider $\gamma\geq 4$: The local minimum $z_-$ of a third order polynomial with negative leading coefficient
satisfies $z_-<z_+$ and $g(z_-)<g(z_+)$. In view of \eqref{g at pip}, it remains to show $g(z_-)>0$. It holds that
$g(z_-)=\frac{3\gamma+(\gamma-4)(2+\gamma+\sqrt {(\gamma-4)(\gamma-1)})}{27 (\gamma - 1) \varepsilon}+O(1)$, 
whence $g(z_-)>0$ for sufficiently small $\varepsilon$. Hence $g>0$ on  $[\pi_+,1/\varepsilon]$ is shown.
\end{proof}

\begin{lemma}\label{existence controllable strategy}
Let $\eta_-<\eta_+$ be such that either $\eta_+<-1/(1-\varepsilon)$ or $\eta_->0$. Then there exists an admissible trading strategy ${\hat\varphi}$ such that the risky-safe ratio $\eta_t$
satisfies SDE \eqref{eq zeta diff}. Moreover, $(\eta_t,{\hat\varphi}_t^\uparrow,{\hat\varphi}_t^\downarrow)$ is a reflected diffusion on
the interval $[\eta_-,\eta_+]$. In particular, $\eta_t$ has stationary density equals
\begin{equation}\label{eq st de 1}
 \nu(\eta):=\frac{\frac{2\mu}{\sigma^2}-1}{\eta_+^{\frac{2\mu}{\sigma^2}-1}-\eta_-^{\frac{2\mu}{\sigma^2}-1}}\eta^{\frac{2\mu}{\sigma^2}-2}, \quad \eta\in[\eta_-,\eta_+],
\end{equation}
when $\eta_->0$, and otherwise equals
\begin{equation}\label{eq st de 2}
 \nu(\eta):=\frac{\frac{2\mu}{\sigma^2}-1}{\vert\eta_-\vert^{\frac{2\mu}{\sigma^2}-1}-\vert\eta_+\vert^{\frac{2\mu}{\sigma^2}-1}}\vert\eta\vert^{\frac{2\mu}{\sigma^2}-2}, \quad \eta\in[\eta_-,\eta_+].
\end{equation}
\end{lemma}
\begin{proof}
By the solution of the Skorohod problem for two reflecting boundaries \citep{kruk2007explicit}, there exists a well-defined reflected diffusion $(\eta_t, L_t, U_t)$ satisfying $
\frac{d\eta_t}{\eta_t}=\mu dt+\sigma dB_t+dL_t-dU_t$, where $B$ is a standard Brownian motion. If $\eta_->0$, $L$ (resp. $U$) is a non-decreasing processes which increases only on the set $\{\eta=\eta_-\}$ (resp. $\{\eta=\eta_+\}$)\footnote{If $\eta_+<0$, the terms ``decreasing'' and ``increasing'' are exchanged.}. Also, $\eta_->0$
or $\eta_+<-1/(1-\varepsilon)$ implies that $\eta_t>0$ or $\eta_t<-1/(1-\varepsilon)$ for all $t$, almost surely. Hence for each $t>0$ the coefficients
$ (1+(1-\varepsilon)\eta_t)$ and $(1+\eta_t)$ are invertible, almost surely. Define the increasing processes $({\hat\varphi} ^\uparrow, {\hat\varphi}^\downarrow)$ by
$\frac{d{\hat\varphi}_t^\uparrow}{{\hat\varphi}_t}=(1+\eta_t)^{-1}dL_t$, $\frac{d{\hat\varphi}_t^\downarrow}{{\hat\varphi}_t}=(1+(1-\varepsilon)\eta_t)^{-1}dU_t$. The associated measures $d{\hat\varphi} ^\uparrow, d{\hat\varphi}^\downarrow $ are supported on $\eta_t=\eta_-$ and $\eta_t=\eta_+$, respectively. Hence ${\hat\varphi}$ is a trading strategy, which by Lemma \ref{le: rewriting obj fun}  yields
a risky-safe satisfying precisely the stochastic differential equation \eqref{eq zeta diff}. The admissibility of the trading strategy is clear, as ${\hat\varphi}$ is a continuous, finite variation trading strategy, and it satisfies $\pi_+ <1/\varepsilon$, which implies that
there exists $\varepsilon'>\varepsilon$ such that $\pi_t<1/\varepsilon'$, for all $t>0$, a.s.. 

Write the infinitesimal generator of $(\eta_t)_{t\geq 0}$ in the general form 
\[
\mathcal Af(\eta)=\frac{\sigma^2}{2}\eta^2 f''(\eta)+\mu\eta f'(\eta)=: \frac{1}{2}a^2(\eta) f''(\eta)+b(\eta)f'(\eta).
\]
The speed measure is $m(d\eta)=\left( \frac{2}{a^2(\eta)}e^{\int _{\eta_-}^\eta \frac{2 b(y)}{a^2(y)}dy}\right) d\eta$, and as $m([\eta_-,\eta_+)]<\infty$, 
$(\eta_t)_{t\geq 0}$ is positively recurrent and its invariant density $\nu$ is 
\begin{equation}
\nu(\eta)d\eta=\frac{m(d\eta)}{m([\eta_-,\eta_+])},
\end{equation}
(see \cite[ II.9, II. 12, and II. 36]{borodin2002handbook}). Distinguishing the cases $\eta_+<0$ or $\eta_->0$, the probability densities \eqref{eq st de 1} and \eqref{eq st de 2} follow.
\end{proof}

The following constitutes the verification of optimality of the trading strategy of Lemma \ref{existence controllable strategy}
with the trading boundaries in Proposition \ref{prop1}:
\begin{proposition}\label{prop3}
Let $\zeta_\pm$ be the free boundaries as derived in Proposition \ref{prop1}, and set $\pi_\pm:=\zeta_\pm/(1+\zeta_\pm).$ Denote by $\hat \varphi$
the trading strategy of Lemma \ref{existence controllable strategy} associated with these free boundaries. Then for all $t>0$, the fraction of wealth $\pi_t$ invested in the risky asset lies in
the interval $[\pi_-,\pi_+]$, almost surely, entails no trading whenever $\pi\in (\pi_-,\pi_+)$ (the no-trade region) and engages in trading
only at the boundaries $\pi_\pm$. For sufficiently small $\varepsilon$, $\hat\varphi$ is optimal, and
the value function is 
\begin{align}\nonumber
\textstyle{F_{\infty}(\hat\varphi)}&\textstyle{=r+\max_{\varphi\in \Phi}\lim_{T\rightarrow\infty}\frac1T 
\mathbb E\left[
\int_0^T \left( \mu  \pi_t - \frac\gamma 2 \sigma^2 \pi_t^2 \right)dt
-{\varepsilon}\int_0^T\pi_t\frac{d\varphi^\downarrow_t}{\varphi_t} 
\right]}\\\label{eq expl value function}&\textstyle{=r+\mu \pi_--\frac{\gamma \sigma^2}{2}\pi_-^2.}
\end{align}
\end{proposition}

\begin{proof}[Proof of Proposition \ref{prop3}]
Recall from Proposition \ref{prop2} that $\lambda=h(\zeta_-)$ and $(V,\lambda)$, defined from
the unique solution of the free boundary problem, is a solution of the HJB equation \eqref{eq: HJBx}. For the verification, the proportion $\pi_t$ of wealth in the risky asset is used instead of the risky-safe ratio $\zeta_t$. The change of variable $
 \zeta=-1+\frac{1}{1-\pi}$ amounts to a compactification of the real line, such that the two intervals $[-\infty, -1/(1-\varepsilon))$ and $(0,\infty]$ are mapped onto
the connected interval $[0,1/\varepsilon)$. Denote by $\mathcal L$ the differential operator
\[
\textstyle{(\mathcal L f)(\pi):=\frac{\sigma^2}{2}f''(\pi)\pi^2(1-\pi)^2+f'(\pi)(\mu-\sigma^2\pi)\pi (1-\pi).}
\]
Set $\hat h(\pi)=h(\zeta(\pi))=\mu \pi-\frac{\gamma \sigma^2}{2}\pi^2$. The function $\hat V(\pi):=V(\zeta(\pi))$ satisfies the HJB equation
\begin{equation}\label{eq HJB new}
\textstyle{\min(\mathcal L \hat V(\pi)-\hat h(\pi)+\lambda, \hat V'(\pi), \varepsilon/(1-\varepsilon \pi)-\hat V'(\pi))=0,\quad  0\leq\pi<1/\varepsilon.}
\end{equation}
First, note that $F_{\infty}(\varphi)\leq \lambda+r$ for any admissible trading strategy $\varphi$: By Lemma \ref{lem integrability} and Remark \ref{rem: range of trade}, without loss of generality assume $\pi_t\geq 0$ almost surely for all $t\ge 0$. An application of It\^o's formula to the stochastic process $\hat V(\pi_t)$, where $\hat V$ is the solution of the HJB equation \eqref{eq HJB new}, yields
\begin{align}\label{ito for veri}
\textstyle{\hat V(\pi_T)-\hat V(\pi_0)}&\textstyle{=\int_0^T \hat V'(\pi_t)d\pi_t+\frac{1}{2}\hat V''(\pi_t) d\langle \pi\rangle_t}\\\label{term 1 in eq}
&\textstyle{=\int_0^T \left(\mathcal L \hat V(\pi)-\hat h(\pi_t)+\lambda\right)dt +\int_0^T (\hat h(\pi_t)-\lambda)dt}\\\label{term 2 in eq}
&\textstyle{+ \int_0^T \hat V'(\pi_t)\pi_t (1-\pi_t)\sigma dB_t}\\\label{term 4 in eq}
&\textstyle{- \int_0^T\hat V '(\pi_t)(1-\varepsilon \pi_t)\pi_t\frac{d\varphi_t^\downarrow}{\varphi_t}}\\
\label{term 3 in eq}
&\textstyle{+ \int_0^T \hat V'(\pi_t)\pi_t \frac{d\varphi_t^\uparrow}{\varphi_t}.}
\end{align}

The first term in line \eqref{term 1 in eq} is non-negative, in view of \eqref{eq HJB new}. Furthermore, \eqref{eq: liquidity} implies the existence of $\varepsilon'>\varepsilon$ such that $\pi_t<1/\varepsilon'<1/\varepsilon$, for all $t$, a.s. Using \eqref{eq HJB new}  one thus obtains
\begin{equation}\label{hat v prime est}
\hat V'(\pi_t)\leq \frac{\varepsilon\varepsilon'}{\varepsilon'-\varepsilon},\qquad\text{a.s. for all }t\ge 0.
\end{equation}
Hence \eqref{term 2 in eq} is a martingale with zero expectation. Again, \eqref{eq HJB new} implies that
\[
\textstyle{\hat V'(\pi_t)\pi_t(1-\varepsilon\pi_t)\leq \varepsilon \pi_t,}
\]
whence \eqref{term 4 in eq} satisfies
\[
\textstyle{- \int_0^T\hat V '(\pi_t)(1-\varepsilon \pi_t)\pi_t\frac{d\varphi_t^\downarrow}{\varphi_t}\geq -\varepsilon \int_0^T \pi_t \frac{d\varphi^\downarrow_t}{\varphi_t}.}
\]
Finally, \eqref{term 3 in eq} is non-negative, because $\hat V'\geq 0$ due to \eqref{eq HJB new}.

Taking the expectation of \eqref{ito for veri} yields the estimate
\begin{equation}\label{est for veri}
\textstyle{\frac{1}{T}\mathbb E[\hat V(\pi_T)-\hat V(\pi_0)]\geq -\lambda+\frac{1}{T}\mathbb E[\int_0^T \hat h(\pi_t)dt]-\varepsilon \frac{1}{T}\int_0^T \pi_t\frac{d\varphi_t^\downarrow}{\varphi_t}.}
\end{equation}
By eq.~\eqref{hat v prime est}
\[
\textstyle{\vert \hat V(\pi_t)-\hat V(\pi_0)\vert\leq \vert \pi_T-\pi_0\vert \sup_{0<u \leq 1/\varepsilon '}\vert \hat V'(u)\vert\leq \frac{\varepsilon}{\varepsilon'-\varepsilon},}
\]
therefore $\lim_{T\rightarrow \infty}\frac{1}{T}\mathbb E[\hat V(\pi_T)-\hat V(\pi_0)]=0$. Hence letting $T\rightarrow \infty$ in \eqref{est for veri} implies that for any admissible strategy $\varphi$ one has $F_{\infty}(\varphi)\leq \lambda+r$. Finally, this bound is attained by the admissible trading strategy $\hat \varphi$
defined by Lemma \eqref{existence controllable strategy} in terms of the free boundaries $(\zeta_-,\zeta_+)$: Let $\zeta_t$ be the corresponding
risky-safe ratio. Using It\^o's formula, one has $
dV(\zeta_t)=V'(\zeta_t)\zeta_t \sigma d B_t+ 0-\varepsilon \pi_t \frac{d\varphi_t^\downarrow}{\varphi_t}+(h(\zeta_t)-\lambda)dt$. Division by $T$ yields, in view of \eqref{eq main problem 1},
\[
\textstyle{\frac1T 
\mathbb E\left[
\int_0^T \left( \mu \pi_t - \frac\gamma 2 \sigma^2 \pi_t^2 \right)dt
-{\varepsilon}\int_0^T\pi_t\frac{d\varphi^\downarrow_t}{\varphi_t} 
\right]=\lambda+\frac{1}{T}\mathbb E[\hat V(\pi_T)-\hat V(\pi_0)].}
\]
Letting $T\rightarrow\infty$, one obtains $F_{\infty}(\hat \varphi)=\lambda+r$. 
\end{proof}
\subsection{Proof of Theorem \ref{thm free b} \ref{thm free b part 1}--\ref{thm free b part 3}}
Theorem \ref{thm free b} \ref{thm free b part 1} is proved in Proposition \ref{prop1}, and
Theorem \ref{thm free b} \ref{thm free b part 2} \& \ref{thm free b part 3} are proved
in Proposition \ref{prop3}.

\section{Performance and Asymptotics}\label{app proof main theorem}
In this section, ergodicity arguments are used to derive closed-form expressions for average trading costs ($\avtrco$) and long-run mean and long-run variance of the optimal trading strategy. These formulae in turn yield the asymptotic expansions of Theorem \ref{thm free b} \ref{thm free b part 4}.

\subsection{The frictionless contribution}
Let $\zeta_-,\zeta_+$ be the free boundaries obtained in Proposition \ref{prop1}. In view of Remark \ref{rem: range of trade}, assume that either $\zeta_-<\zeta_+<-1$ (leveraged case)
or $\zeta_->\zeta_+>0$ throughout (non-leveraged case), and define the integral
\begin{equation}\label{eq: obj fricless}
I:=\frac{1}c\int_{\zeta_-}^{\zeta_+} h(\zeta)\vert\zeta\vert^{2\gamma\pi_*-2}d\zeta,
\end{equation}
where the normalizing constant is 
\begin{equation}\label{const c}
c:=\int_{\zeta_-}^{\zeta_+}\vert \zeta\vert^{2\gamma\pi_*-2}d\zeta= \sgn(\zeta_-)\frac{\vert\zeta_+\vert^{2\gamma\pi_*-1}-\vert\zeta_-\vert^{2\gamma\pi_*-1}}{2\gamma\pi_*-1}.
\end{equation}

\begin{lemma}\label{integral h}
\begin{equation}\label{eq: normalised mean var}
I=h(\zeta_-)+\frac{\sigma^2(2\gamma\pi_*-1)}{2} \left(\frac{G({\zeta_+}) {\zeta_+}}{1-\left(\frac{{\zeta_-}}{{\zeta_+}}\right)^{2\gamma\pi_*-1}}\right).
\end{equation}
\end{lemma}
\begin{proof}
From equations \eqref{TKA IVP sol} and \eqref{eq: TAK1} it follows that
\[
 \int_{\zeta_-}^{\zeta_+} h(\zeta)\vert\zeta\vert^{2\gamma\pi_*-2}d\zeta=h({\zeta_-})\sgn(\zeta_-)\frac{{\vert\zeta_+\vert}^{2\gamma\pi_*-1}-{\vert\zeta_-\vert}^{2\gamma\pi_*-1}}{2\gamma\pi_*-1}+\frac{\sigma^2 {\zeta_+}^{2\gamma\pi_*}}{2}G({\zeta_+}).
\]
By normalizing, \eqref{eq: normalised mean var} follows.
\end{proof}
\subsection{Transaction costs}
For the optimal trading policy, the risky-safe ratio $\zeta$ is a geometric Brownian motion with parameters $(\mu,\sigma)$, reflected at $\zeta_-,\zeta_+$ respectively, see Lemma \ref{existence controllable strategy}. Hence the following ergodic result \cite[Lemma C.1]{gerhold.al.11} applies:
\begin{lemma}\label{lem ergodic}
Let $\eta_t$ be a diffusion on an interval $[l,u]$, $0<l<u$, reflected at the boundaries, i.e.
\[
d\eta_t=b(\eta_t)dt+a(\eta_t)^{1/2}dB_t +d  L_t-d U_t,
\]
where the mappings $a(\eta)>0$ and $b(\eta)$ are both continuous, and the continuous, non-decreasing
processes $  L_t$ and $  U_t$ satisfy $  L_0=   U_0=0$ and increase only on $\{  L_t=l\}$
and $\{  U_t=u\}$, respectively. Denoting by $\nu(\eta)$ the invariant density of $\eta_t$, the following almost sure limits hold:
\[
\lim_{T\rightarrow \infty} \frac{   L_T}{T}=\frac{a(l) \nu(l)}{2},\quad \lim_{T\rightarrow \infty} \frac{   U_T}{T}=\frac{a(u) \nu(u)}{2}.
\]
\end{lemma}
The next formula evaluates trading costs.
\begin{lemma}\label{Lem atc}
The average trading costs for the optimal trading policy are 
\begin{equation}\label{eq: ATC appendix}
\avtrco:=\varepsilon\lim_{T\rightarrow\infty}\frac{1}{T}\int_0^T \pi_t \frac{d\varphi_t^\downarrow}{\varphi_t}=\frac{\sigma^2(2\gamma\pi_*-1)}{2} \left(\frac{G({\zeta_+}) {\zeta_+}}{1-\left(\frac{{\zeta_-}}{{\zeta_+}}\right)^{2\gamma\pi_*-1}}\right).
\end{equation}
\end{lemma}
\begin{proof}
Note that $\varepsilon \int_0^T \pi_t \frac{d\varphi_t^\downarrow}{\varphi_t}=G({\zeta_+})\frac{U_T}{T}$. Applying Lemma \ref{lem ergodic} to $\eta:=\zeta$ (setting $l:=\zeta_-$, $u={\zeta_+}$) and using the stationary density of $\zeta_t$ (Lemma \ref{existence controllable strategy}) \eqref{eq: ATC appendix} follows.
\end{proof}
\begin{remark}\rm
An alternative proof of Lemma \ref{Lem atc} follows from Lemma \ref{le: rewriting obj fun}, by rewriting the objective functional as $
F_\infty(\varphi)=r+\lim_{T\rightarrow \infty}\frac{1}{T}\int_0^T h(\zeta_t)dt - \avtrco$.
By the ergodic theorem \cite[II.35 and II.36]{borodin2002handbook}, $\lim_{T\rightarrow \infty}\frac{1}{T}\int_0^T h(\zeta_t)dt=I$
hence using Lemma \ref{integral h} and Proposition \ref{prop3} it follows that
\[
\avtrco=-F_\infty(\varphi)+r+I=\frac{\sigma^2(2\gamma\pi_*-1)}{2} \left(\frac{G({\zeta_+}) {\zeta_+}}{1-\left(\frac{{\zeta_-}}{{\zeta_+}}\right)^{2\gamma\pi_*-1}}\right).
\]
which is in agreement with the formula in Proposition \ref{prop3}.
\end{remark}

\subsection{Long-run mean and variance}\label{sec: long-run mean and variance}
Set
\[ 
I_{\mu}:=\int_{\zeta_-}^{\zeta_+} \left(\frac{\zeta}{1+\zeta}\right)\vert\zeta\vert^{2\gamma\pi_*-2}d\zeta,\quad 
I_{s^2}:=\int_{\zeta_-}^{\zeta_+} \left(\frac{\zeta}{1+\zeta}\right)^2\vert\zeta\vert^{2\gamma\pi_*-2}d\zeta.
\]
 In view of the ergodic theorem \cite[II.35 and II.36]{borodin2002handbook}, the long-run mean and long-run variance satisfy
\begin{align*}
\hat m&=r+\mu\lim_{T\rightarrow\infty}\frac{1}{T}\mathbb E[\int_0^T \pi_t dt]-\avtrco=r+\frac{\mu}{c} I_\mu-\avtrco,\\
\hat s^2&=\sigma^2\lim_{T\rightarrow\infty}\frac{1}{T}\mathbb E[\int_0^T \pi_t^2 dt]=\frac{\sigma^2}{c}I_{s^2},
\end{align*}
whence the following decomposition holds:
\begin{equation}\label{I decomp}
I=\frac{1}{c}\left(\mu I_{\mu}-\frac{\gamma  \sigma^2}{2}I_{s^2}\right)=\frac{\pi_*}{\pi_*}(\hat m-r+\avtrco)-\frac{\gamma}{2}\hat{\sigma^2}=h(\zeta_-)+\avtrco.
\end{equation}
Integration by parts yields
\begin{equation}\label{eq IBP mu}
I_{\mu}=\int_{\zeta_-}^{\zeta_+} \frac{\zeta}{1+\zeta}\vert \zeta\vert^{2\gamma\pi_*-2}d\zeta=\frac{{\vert\zeta_+\vert}^{2\gamma\pi_*}}{2\gamma\pi_*(1+{\zeta_+})}-\frac{{\vert\zeta_-\vert}^{2\gamma\pi_*}}{2\gamma\pi_*(1+{\zeta_-})}+\frac{I_{s^2}}{2\gamma\pi_*}.
\end{equation}
Plugging \eqref{eq IBP mu} into \eqref{I decomp} yields $I=\frac{\sigma^2}{2c}\frac{\pi_*}{\pi_*}\left(\frac{\vert\zeta_+\vert^{2\gamma\pi_*}}{1+\zeta_+}-\frac{\vert\zeta_-\vert^{2\gamma\pi_*}}{1+\zeta_-}+(1-\frac{\gamma\pi_*}{\pi_*}) I_{s^2}\right)$. Except for the singular case $\gamma=1$, one can extract $I_{s^2}$, and thus \eqref{eq IBP mu} and \eqref{eq: ATC appendix} yield a formula for
$\hat{s}^2$. Therefore, the right side of equation \eqref{I decomp} gives a formula for $\hat m$ in terms of $\hat s$:
\begin{lemma}\label{Lemma long run mean and var}

When $\gamma\neq 1$, the following identities hold:
\begin{align}\label{eq long run var}
\hat s^2&=\frac{2}{1-\gamma}\left(h(\zeta_-)+\avtrco \right)-\frac{\sigma^2}{c(1-\gamma) }\left(\frac{\vert\zeta_+\vert^{2\gamma\pi_*}}{1+\zeta_+}-\frac{\vert\zeta_-\vert^{2\gamma\pi_*}}{1+\zeta_-}\right),\\\label{eq long run mean explicit}
\hat m&=r+\frac{\gamma}{2}\hat s^2 +h(\zeta_-).
\end{align}

\end{lemma}

\subsection{Proof of Theorem \ref{thm free b} \ref{thm free b part 4}}\label{as of th31}
\begin{proof}
The asymptotic expansion \eqref{eq: trading boundaries TKA x} for the trading boundaries $\pi_\pm$ is derived by expanding $\frac{\zeta_\pm}{1+\zeta_\pm}$ into a power series, thereby using the asymptotic expansions \eqref{eq: asymptotics zeta} of $\zeta_\pm$.

The long-run mean $\hat m$, variance $\hat s^2$, Sharpe ratio ($(\hat m - r)/\hat s$), average trading costs $\avtrco$, and value function $\lambda$
have closed form expressions in terms of the free boundaries $\zeta_-,\zeta_+$ (see equations \eqref{eq long run mean explicit}, \eqref{eq long run var}, and equations
\eqref{eq: ATC appendix} and \eqref{eq expl value function}). Using these formulas in combination with the asymptotic expansions \eqref{eq: asymptotics zeta} of the free boundaries,
the assertion follows.
\end{proof}

\section{From Risk Aversion to Risk Neutrality}\label{sec: proof th multiplier}
In this section the free boundary problem \eqref{eq: TKA fbp}--\eqref{terminal1 TKA} for $\gamma=0$ is solved for sufficiently small $\varepsilon$, it is shown
that the solution $(W,\zeta_-,\zeta_+)$ allows to construct a solution of the corresponding HJB equation and, similarly to the case $\gamma>0$, a verification argument yields the strategy's optimality.

Numerical experiments using $\gamma>0$ indicate that the trading boundaries $\pi_\pm$ (hence the leverage multiplier) satisfy $\lim_{\varepsilon\downarrow 0}\varepsilon^{1/2}\pi_\pm=1/A_\pm$ for two constants $A_->A_+>0$. This entails that the free boundaries have the approximation $\zeta_\pm\approx -1-A_\pm \varepsilon^{1/2}$, thereby suggesting that $\zeta_\pm$ are analytic in $\delta:=\varepsilon^{1/2}$. The system \eqref{eq: TAK1}--\eqref{eq: TAK2} is rewritten by using the new parameter $\delta:=\varepsilon^{1/2}$ and by multiplying the second equation by $\delta$:
\begin{align}\label{eq: TAK1half}
\textstyle{W(\zeta_-,\zeta_+)-\frac{\delta^2}{(1+\zeta_+)(1+(1-\delta^2)\zeta_+)}}&=0,\\
\label{eq: TAK2half}
\textstyle{\delta\left(\frac{2 (h(\zeta_+)-h(\zeta_-))}{\sigma^2 \zeta_+^2}-\frac{2\mu/\sigma^2}{\zeta_+} W(\zeta_-,\zeta_+)-\frac{(1-\delta^2)^2}{(1+(1-\delta^2)\zeta_+))^2}+\frac{1}{(1+\zeta_+)^2}\right)}&=0.
\end{align}
Using the transformation $u=\frac{-1-\zeta}{\delta}$ and noting that $\vert \zeta\vert=1+\delta u$, it follows that 
\[
\textstyle{\Xi (u_-,u):=W(-1-u_-\delta,-1-u\delta)=\frac{2 \mu}{\sigma^2(1+u\delta)^2}\int_{u_-}^u \left(\frac{1}{u_-}-\frac{1}{\xi}\right)\left(\frac{1+\xi \delta}{1+u\delta}\right)^{\frac{2\mu}{\sigma^2}-2}d\xi.}
\]
Accordingly, the system \eqref{eq: TAK1half}--\eqref{eq: TAK2half} transforms into
\begin{align}\label{fbp u1}
\textstyle{\Xi(u_-,u_+)-\frac{1}{u_+((1-\delta^2)u_+-\delta)}}&=0,\\
\label{fbp u2}
\textstyle{\frac{2\mu}{\sigma^2}\left(\frac{1}{u_+}-\frac{1}{u_-}+\frac{\delta }{1+u_+\delta} \Xi (u_-,u_+)\right)-\frac{2 \left(1-\delta ^2\right) u_+-\delta}{ u_+^2 \left(\delta
   +\left(\delta ^2-1\right) u_+\right)^2}}&=0.
\end{align}
Letting $\delta\rightarrow 0$ in \eqref{fbp u1}--\eqref{fbp u2}, one obtains a system of equations for $(A_-,A_+)$,
\begin{equation}\label{first order eq1}
\frac{2\mu}{\sigma^2}\left(\log(A_-/A_+)-\frac{A_--A_+}{A_-}\right)-\frac{1}{A_+^2}=0,\quad
\frac{\mu}{\sigma^2}\left(\frac{1}{A_+}-\frac{1}{A_-}\right)-\frac{1}{A_+^3}=0.
\end{equation}

\begin{lemma}\label{prop super}
The unique solution $(A_-,A_+)$ of the system \eqref{first order eq1} is
\begin{equation}\label{eq: first order proxies}
A_-=\frac{\kappa^{-1/2}}{1-\kappa}\sqrt{\frac{\sigma^2}{\mu}},\quad A_+=\kappa^{-1/2}\sqrt{\frac{\sigma^2}{\mu}},
\end{equation}
where $\kappa\approx 0.5828$ is the unique solution of \eqref{unicorn dimless}.
\end{lemma}
\begin{proof}
The second equation in \eqref{first order eq1} gives
\begin{equation}\label{eqam}
\textstyle{A_-=\frac{\mu A_+^3}{\mu A_+^2-\sigma^2}.}
\end{equation}
Hence substituting \eqref{eqam} into the first equation of \eqref{first order eq1} gives the well-posed problem
\begin{equation}\label{eq unicorn}
\textstyle{-\frac{3}{A_+^2}+\frac{2\mu \log\left(\frac{\mu A_+^2}{\mu A_+^2-\sigma^2}\right)}{\sigma^2}=0,\quad A_+>0.}
\end{equation}
Therefore it is enough to establish that the unique solution of \eqref{eq unicorn} is as in the second equation in line \eqref{eq: first order proxies}; the formula for $A_-$ then follows from \eqref{eqam}. To this end, substitute $\xi:=\sigma^2/(\mu A_+^2)$ into \eqref{eq unicorn} to obtain equation \eqref{unicorn dimless}. Note that $f(0)=0$, $f'>0$ on $(0,1/3)$ and $f'<0$ on $(1/3,1)$, while $f(\xi)\downarrow-\infty$ as $\xi\rightarrow 1$. This implies that $f$ has a single zero $\kappa$ on $(1/3,1)$
and thus the claim concerning $A_+$ is proved.
\end{proof}
\begin{proposition}\label{prop fbp at infinity}
For sufficiently small $\delta$, there exists a unique solution $(u_+,u_-)$ of \eqref{fbp u1}--\eqref{fbp u2} near $(A_-,A_+)$. This solution is analytic in $\delta$ and satisfies
the asymptotic expansion $u_\pm=A_\pm +O(\delta)$, where $A_\pm$ are in \eqref{eq: first order proxies}.
\end{proposition}
\begin{proof}
Denote the left sides of \eqref{fbp u1}--\eqref{fbp u2}, by $F_i((u_-,u_+),\delta)$, $i=1,2$ and $F=(F_1,F_2)$.
By Lemma \ref{prop super}, $F((A_-,A_+),0)=0$. As
\[
\textstyle{\frac{\partial \Xi}{\partial u_-}((A_-,A_+),0)=\frac{2\mu}{\sigma^2}\left(\frac{A_--A_+}{A_-^2}\right),\quad \frac{\partial \Xi}{\partial u_+}((A_-,A_+),0)=\frac{2\mu}{\sigma^2}\left(\frac{A_+-A_-}{A_- A_+}\right),}
\]
one obtains at $(A_\pm)$, $\frac{\partial F_1}{\partial u_-}=\frac{2\mu}{\sigma^2}\left(\frac{A_--A_+}{A_-^2}\right)$, $\frac{\partial F_2}{\partial u_+}=\frac{6}{A_+^4}-\frac{2\mu}{\sigma^2 A_+^2}$
and
\[
\textstyle{\frac{\partial F_1}{\partial u_+}=\frac{2}{A_+^3}+\frac{2\mu}{\sigma^2}\left(\frac{A_+-A_-}{A_-A_+}\right)=0,}
\]
where the last equality follows from the second equation in \eqref{first order eq1}. Therefore, as $\kappa\in(1/3,1)$, the Jacobian $DF$ satisfies
$\det(D F)((A_-,A_+),0)=-4 (\mu/\sigma^2) ^{7/2} (\kappa -1) \kappa ^{5/2} (3 \kappa -1)\neq 0$. By the implicit function theorem for analytic functions  \cite[Theorem I.B.4]{GunningRossi} the assertion follows.
\end{proof}
\begin{lemma}\label{lem Hilfssatz}
Let $\kappa$ be the solution of \eqref{unicorn dimless} and $\theta\in[0,1]$. Then
\begin{equation}\label{eq junge}
\textstyle{\log(1-\kappa(1-\theta))+(1-\theta)\kappa+\frac{1}{2}\frac{\kappa (1-\kappa)^2}{(1-\kappa(1-\theta))^2}=0}
\end{equation}
implies $\theta=0$.
\end{lemma}
\begin{proof}
Clearly $f(0)=0$ and also $f(1)=1/2 \kappa(1-\kappa)^2>0$. There is a single local extremum
of $f$, in $(0,1)$, namely, $
\theta_1=\frac{0.5 \left(3. \kappa ^2+\sqrt{4. \kappa ^3-3. \kappa ^4}-2. \kappa
   \right)}{\kappa ^2}\approx 0.7669$
. Because $f'(0)=0$ and $
f''(0)=\frac{\kappa ^2 \left(\kappa  \left(3 \kappa ^2-7 \kappa
   +5\right)-1\right)}{(1-\kappa )^4}>0$, 
$\theta_1$ must be the global maximum. Hence $f>0$ on $(0,1]$, whence $\theta=0$, as claimed.
\end{proof}
\begin{lemma}\label{lem super kueken}
Let $A_-$ be as in \eqref{eq: first order proxies}. The only solution of
\begin{equation}\label{eq: no 1}
\frac{2\mu}{\sigma^2}\left(\log(A_-/\xi)-\frac{A_--\xi}{A_-}\right)-\frac{1}{\xi^2}=0
\end{equation}
on $[A_+, A_-]$ is $\xi=A_+$.
\end{lemma}
\begin{proof}
Let $\xi$ be a solution of \eqref{eq: no 1}. There exists $\theta\in[0,1]$ such that
\[
\textstyle{\xi=\theta A_-+(1-\theta)A_+=A_+\left(\frac{1+\kappa(1-\theta)}{1-\kappa}\right).}
\]
Hence $A_+^*/A_-=1+\kappa(\theta-1)$, and therefore \eqref{eq: no 1}
is rewritten as \eqref{eq junge}. An application of Lemma \ref{lem Hilfssatz} yields $\xi=A_+$.
\end{proof}

\subsection{Proof of Theorem \ref{th multiplier}}

\begin{proof}
Arguing similarly as in the Proof of Proposition \ref{prop1} for the case $\gamma>0$, the solvability of the free boundary problem \eqref{eq: TKA fbp}--\eqref{terminal1 TKA} for $\gamma=0$ is equivalent to solvability of the non-linear
system \eqref{eq: TAK1half}--\eqref{eq: TAK2half}. This, in turn, is equivalent to solving the system \eqref{fbp u1}--\eqref{fbp u2} for $(u_+(\delta), u_-(\delta))$. A unique solutions of the transformed system \eqref{fbp u1}--\eqref{fbp u2} near
$(A_+,A_-)$ is provided by Proposition \ref{prop fbp at infinity}, and one has $\zeta_\pm=-1-u_\pm \delta$. In particular, one obtains
\begin{equation}\label{eq: zeta riskneutral}
\zeta_\pm=-1-A_\pm \varepsilon^{1/2}+O(1).
\end{equation}
The solution of \eqref{eq: TKA fbp}--\eqref{terminal1 TKA} is 
\begin{equation}\label{eq explicit sol}
W(\zeta):=\frac{2\mu}{\sigma^2 \vert\zeta\vert^{\frac{2\mu}{\sigma^2}}}\int_{\zeta_-}^\zeta \left(\frac{y }{1+y}-\frac{\zeta_-}{1+\zeta_-}\right)\vert y\vert^{2\mu/\sigma^2-2}dy.
\end{equation}
One defines exactly as in \eqref{eq candidate HJB} a candidate solution $(V,\lambda)$ of the HJB equation \eqref{eq: HJBx}.
Next it is shown that $(V,\lambda)$ solves the HJB equation  \eqref{eq: HJBx} (for the intervals $[\zeta_-,\zeta_+]$, $(-\infty, \zeta_-]$ and finally for $[\zeta_+,\infty)$). In fact, the interval $[-1/(1-\varepsilon),0)$ is excluded.

On $[\zeta_-,\zeta_+]$, 
\[
\textstyle{(\mathcal A V(\zeta)-h(\zeta)+h(\zeta_-))'=\frac{1}{2}\sigma^2 \zeta^2 W''(\zeta)+(\sigma^2+\mu)\zeta W'(\zeta)+\mu W(\zeta)-\frac{\mu}{(1+\zeta)^2}=0}
\]
by construction. Because of the initial conditions \eqref{initial0 TKA}--\eqref{initial1 TKA}, $
(\mathcal A V(\zeta)-h(\zeta)+h(\zeta_-))\mid_{\zeta=\zeta_-}=\mathcal A V(\zeta)\mid_{\zeta=\zeta_-}=0$, and thus $\mathcal A V(\zeta)-h(\zeta)+h(\zeta_-)\equiv 0$ for $\zeta\in [\zeta_-,\zeta_+]$. Next it is shown that $0\leq V'\leq G$ on all of $[\zeta_-,\zeta_+]$. As $
(h(\zeta)-h(\zeta_-))'=h'(\zeta)=\frac{\mu}{(1+\zeta)^2}$ is strictly positive, $h(\zeta)-h(\zeta_-)>0$ for $\zeta\in(\zeta_-,\zeta_+]$.
From the explicit formula \eqref{eq explicit sol} it then follows that $V'=W\geq 0$ for $\zeta\in [\zeta_-,\zeta_+]$. It remains to show $V'\leq G$. As $V'(\zeta_+)-G(\zeta_+)=0$, 
and $V'(\zeta_-)-G(\zeta_-)=-G(\zeta_-)<0$, it suffices to rule out any zero $\zeta_+^*$ of $V'(\zeta)-G(\zeta)$ on $(\zeta_-,\zeta_+)$, for sufficiently small $\varepsilon$. This is equivalent to ruling out any zeros of 
\[
 \textstyle{\kappa(u,\delta):=V'(\zeta(u))-G(\zeta(u)), \quad u\in (u_+(\delta),u_-(\delta)),}
\]
where  $\zeta(u)=-1-u\delta$, for sufficiently small $\delta$. Recall that $u_\pm(\delta)$ is implicitly defined by $\zeta_\pm=-1-u_\pm(\delta) \delta$, 
$\lim_{\delta\rightarrow 0}u_\pm (\delta)=A_\pm$. Suppose, by contradiction, that there exists
$\delta_k\downarrow 0$ and a sequence $u_+(\delta_k)$ satisfying $u_-(\delta_k)<u_+^*(\delta_k)<u_+(\delta_k)$ which is a solution of $\kappa(u_+^*(\delta_k),\delta_k)=0$ for each $k\in\mathbb N$. By taking a subsequence, if necessary, one may without loss of generality assume $u_+^*(\delta_k)\rightarrow A_+^*\in [A_+,A_-]$ as $k\rightarrow\infty$. 
Suppose first that $A_+^*=A_+$ and define the map $\delta\mapsto u^*(\delta)$ by intertwining
$u_+$ and $u^*_+$ as follows:
\[
\textstyle{u_+^*(\delta)=\begin{cases}u_+^*(\delta_k),\quad k\in\mathbb N\\u_+(\delta),\quad \delta\neq \delta_k\end{cases}.}
\]
Then for sufficiently small $\delta$, the pair $(u_-(\delta),u_+^*(\delta))$ solves \eqref{fbp u1}--\eqref{fbp u2} near $(A_-,A_+)$, hence by Proposition \ref{prop fbp at infinity}, $u_+^*=u_+$, in contradiction to our previous assumption $\zeta_+^*\in (\zeta_-,\zeta_+)$.  Second, consider the case $A_+^*\in (A_+,A_-]$: By equation \eqref{fbp u1} 
\[
\frac{2\mu}{\sigma^2}\left(\log(A_-/A_+^*)-\frac{A_--A_+^*}{A_-}\right)-\frac{1}{(A_+^*)^2}=0.
\]
Lemma \ref{lem super kueken} states $A_+^*=A_+$, which is also impossible. Hence $V'(\zeta)-G(\zeta)$ has no zeroes on $(\zeta_-,\zeta_+)$, and thus $V$ solves the HJB equation on $[\zeta_-,\zeta_+]$.

Consider now $\zeta\leq \zeta_-$. $V$ solves the HJB equation, if 
\[
\mathcal A V-h(\zeta)+h(\zeta_-)=h(\zeta_-)-h(\zeta)\geq 0, \quad G(\zeta)\geq 0.
\]
The first inequality is clearly fulfilled. Also, as $\zeta<-1/(1-\varepsilon)$ or $\zeta>0$, $G$ is a strictly positive function  on [$-\infty,\zeta_-]$,
which finishes the proof for $\zeta\leq \zeta_-$.

Finally, consider $\zeta\geq \zeta_+$. As $G=W$, it suffices to show that
\begin{equation}\label{eq: eq hjb partial 3 gamma0}
L(\zeta):=\mathcal A V(\zeta)-h(\zeta)+h(\zeta_-)\geq 0,\quad G(\zeta)\geq 0.
\end{equation}
The second inequality has now been proved, and it remains to establish the first inequality in \eqref{eq: eq hjb partial 3 gamma0}. As $h_1(\zeta)=\mu\frac{\zeta}{1+\zeta}-\frac{\sigma^2}{2}\left(\frac{\zeta}{1+\zeta}\right)^2$ it follows that
\[
L(\zeta)=\frac{\sigma^2 \zeta^2}{2}G'(\zeta)+\mu \zeta G(\zeta)-h(\zeta)+h(\zeta_-)=h(\zeta_-)-h_1((1-\varepsilon)\zeta)-\frac{\sigma^2}{2}\left(\frac{\zeta}{1+\zeta}\right)^2.
\]
Therefore, by the boundary conditions at $\zeta_+$, and as $W$ solves the free boundary problem
on $[\zeta_-,\zeta_+]$, 
\[
L(\zeta_+)=\frac{\sigma^2 \zeta^2}{2}W'(\zeta_+)+\mu \zeta W(\zeta_+)+h(\zeta_-)-h(\zeta_+)=0.
\]

To show that $L(\zeta)\geq 0$ for all
$\zeta$, it suffices to show that there are no solutions of the equation $L(\zeta)=0$ on $\zeta\geq \zeta_+$ except $\zeta_+$. The transformation $z=\frac{\zeta}{1+\zeta}$ introduces $F(z,\varepsilon):=L(\zeta(z))$. As $F(\pi_+)=0$, polynomial division by $(z-\pi_+)$ yields \eqref{eq: Dec F}, where the third order polynomial $g$ has derivative
$g'=a_0+a_1z+a_2 z^2$ with certain, relatively complex but explicit coefficients $a_0,a_1,a_2$. By the second formula of \eqref{eq: first order proxies}
\begin{equation}\label{eq: pos g}
g(\pi_+)=-\mu+\frac{3 \sigma^2}{A_+^2}+O(\varepsilon^{1/2})
\end{equation}
is strictly positive for sufficiently small $\varepsilon$ because $\kappa>1/3$. The zeros $z_\pm$ of $g'$
are $z_-=-\frac{1}{2 A_+\varepsilon^{1/2}}+O(1)$, $z_+=\frac{4}{3\varepsilon}+O(1)$. For sufficiently small $\varepsilon$ the first one
is negative, and the second solution is larger than $1/\varepsilon$, hence both are irrelevant. Also, $g'(1/\varepsilon)=\sigma^2/2+O(\varepsilon^{1/2})$ and thus $g'(z)>0$ on all of $[\pi_+,1/\varepsilon]$. Together with \eqref{eq: pos g}, it follows that $g>0$ on $[\pi_+,1/\varepsilon]$. Hence $F(z)>0$ for all $z>\pi_+$ which proves that $(V,\lambda)$ solves the HJB equation \eqref{eq: HJBx}.

Using the proof of Proposition \ref{prop3}, one can obtain assertion \ref{th multiplier issue 2} and \ref{th multiplier issue 3}. Finally, the expansions of the trading boundaries claimed in 
\ref{th multiplier issue 4} follow from the asymptotic expansions of the free boundaries $\zeta_-,\zeta_+$ in \eqref{eq: zeta riskneutral}.

\end{proof}

\section{Convergence}\label{sec proofs convergence}

\begin{lemma}\label{outperform}
Let $\mu>\sigma^2$. There exists $\delta_0>0$ such that for all $\delta\leq \delta_0$ and for all $0\leq \gamma\leq\gamma_0:=\frac{\mu}{\sigma^2}$, the objective functional for a trading strategy $\varphi$ which only engages in buying at $\pi_-=1+\delta$ and selling at $\pi_+=(1-\delta)/\varepsilon>\pi_-$ outperforms a buy and hold strategy. More precisely, 
for all $\gamma\leq\gamma_0$ and for all $\delta\leq \delta_0$ 
\[
F_\infty(\varphi)\geq r+\mu-\frac{\gamma \sigma^2}{2} +\left(\frac{\mu-\gamma\sigma^2}{2} \right)\delta .
\]
\end{lemma}
\begin{proof}
As $\varepsilon\in (0,1)$ and $\widetilde\pi:=\frac{\mu}{\gamma_0\sigma^2}>1$, there exists $\widetilde \delta>0$ such that $\pi_*\geq \pi_+$ for all $\delta\leq \widetilde \delta$ and 
$\gamma\leq \gamma_0$.

Let $\rho(\pi)d\pi=\nu(\pi/(1-\pi))\frac{d\pi}{(\pi-1)^2}$, where $\nu(d\zeta)$ is the stationary density of a reflected diffusion $\zeta$ on $[\zeta_-,\zeta_+]$ (Lemma \ref{existence controllable strategy}).  As  $\pi_*\geq \pi_+$, also $\mu\pi-\frac{\gamma\sigma^2}{2}\pi^2\geq \mu\pi_--\frac{\gamma\sigma^2}{2}\pi_-^2$ holds for all $\pi\in [\pi_-,\pi_+]$. Thus, 
\begin{align}\nonumber
\textstyle{F_\infty(\varphi)}&\textstyle{=  r+\int_{\pi_-}^{\pi_+}\left( \mu\pi -\frac{\gamma \sigma^2}{2}\pi^2\right)\rho(d\pi)-\avtrco} \\\nonumber&\textstyle{\geq r+ \mu(1+\delta)-\frac{\gamma\sigma^2}{2}(1+\delta)^2-\frac{(\delta +1) (2 \epsilon -1)^3 \left(2 \mu -\sigma ^2\right)}{4 \epsilon  \left(\delta  \left(\frac{-2 (\delta +1) \epsilon +\delta +1}{\delta }\right)^{\frac{2 \mu }{\sigma ^2}}+(\delta +1) (2 \epsilon -1)\right)}}\label{eq: estimate line}\\&\textstyle{\geq r+\mu-\frac{\gamma\sigma^2}{2} + (\mu-\gamma \sigma^2)\delta-O(\delta^{\min(2,\frac{2\mu}{\sigma^2}-1)}),}
\end{align}
where Lemma \ref{Lem atc} has been invoked to calculate and estimate the average trading costs $\avtrco$. The asymptotic expansion holds for sufficiently small $\delta$ and , as
$\mu>\gamma \sigma^2$, the exponent in the asymptotic formula \eqref{eq: estimate line} satisfies $2\mu/\sigma^2-2>1$.
\end{proof}

\subsection{Proof of Theorem \ref{lem convergence}}
\begin{proof}
As $\zeta_+<-1/(1-\varepsilon)$, the curves $(0,\bar\gamma]\rightarrow \mathbb R:\gamma\mapsto\pi_{\pm}(\gamma)$ range in a relatively compact set, namely $[1,\frac{1}{\varepsilon})$. Consider therefore a sequence
$\gamma_k$, $k=1,2,\dots$ which satisfies $
1\leq \pi_-^0:= \lim_{i\rightarrow\infty}\pi_-(\gamma_k)\leq \lim_{i\rightarrow\infty}\pi_+(\gamma_k)=:\pi_+^0\leq 1/\varepsilon$. Set $\zeta^k_{\pm}:=\frac{\pi_\pm(\gamma_k)}{1-\pi_\pm(\gamma_k)}$, for $k=0,1,2,\dots$, and note that $-\infty\leq \zeta_-^0\leq \zeta_+^0\leq - \frac{1}{1-\varepsilon}$. 
For each $k$, $k=1,2,\dots$, by assumption the HJB equation \eqref{eq: HJBx} is satisfied with $\lambda=\lambda_k:=h(\zeta_-^k)$. The verification arguments in the proof of Proposition \ref{prop3} yield that the trading strategies associated with the intervals $[\pi_-(\gamma_k),\pi_+(\gamma_k)]$
are optimal.

Next, three facts are proved. First $\pi_-^0>1$, which is equivalent to $\zeta_-^0>-\infty$. Suppose, by contradiction, that $\pi_-^0=1$. Then 
$\pi_-(\gamma_k)\rightarrow 1$ and thus $\lambda_k\rightarrow \mu$, as $k\rightarrow\infty$. Hence,
the objective functional eventually minorizes the uniform bound provided by Lemma \ref{outperform}, a mere impossibility to optimality. Hence $\pi_-^0>1$. Second, $\pi_-^0<\pi_+^0$: This holds due to the fact that, by observing limits for the initial and terminal conditions
of zero order in \eqref{eq: TKA fbp}, $W(\zeta_-^0)=0<G(\zeta_-^0)$.
Third, $\pi_+^0<\frac{1}{\varepsilon}$. Suppose, by contradiction, that $\pi_+^0=\frac{1}{\varepsilon}$. Then $G(\zeta_+^k)\rightarrow \infty$, as $k\rightarrow \infty$,
and, as $\zeta_-^0<\zeta_+^0$,  the average trading costs corresponding to $\gamma_k$ satisfy (by Lemma \ref{Lem atc})
\[
\textstyle{\avtrco(k):=\frac{\sigma^2\left(\frac{2\mu}{\sigma^2}-1\right)}{2}\frac{G(\zeta_+^k)\zeta_+^k}{1-\left(\frac{\zeta_-^k}{\zeta_+^k}\right)^{2\mu/\sigma^2-1}}\rightarrow \infty,}
\]
as $k\rightarrow\infty$. Denote by $\hat\varphi^k$ the trading strategy which only buys (resp. sells) ath $\pi_-(\gamma_k)$ (resp. $\pi_+(\gamma_k)$). By the results of Appendix \ref{app proof main theorem} the value function satisfies 
\[
\lim_{k\rightarrow\infty}{F_{\infty}(\hat \varphi^k)=
\lim_{k\rightarrow\infty}\int\limits_{\pi_-(\gamma_k)}^{\pi_+(\gamma_k)}(\mu\pi-\frac{\gamma_k \sigma^2}{2}\pi^2)\rho(d\pi)-\avtrco(k)\leq \frac{\mu}{\varepsilon}-\lim_{k\rightarrow\infty}\avtrco(k)= -\infty}
\]
as $k\rightarrow \infty$. In particular, for sufficiently large $k\geq k_0$, a buy-and-hold strategy $\varphi$ satisfies $F_{\infty}(\varphi)=\mu-\frac{\gamma_k\sigma^2}{2}>F_{\infty}(\hat \varphi^k)$,  which contradicts optimality of the trading strategy $[\pi_-(\gamma_k),\pi_+(\gamma_k)]$. Hence $\pi_+^0<1/\varepsilon$.

As the sequence $\zeta_-^k$ converges, by \cite[Lemma 9]{keller2010convexity} the solutions of the initial value problem associated with
\eqref{eq: TKA fbp} and $\gamma_{k}$, namely $W(\zeta; \zeta_-^k)$, converge to the solution of the initial value problem \eqref{eq: TKA fbp} (for $\gamma=0$), 
\[
W^0(\zeta)=-\frac{2}{\sigma^2\zeta^{2}}\int_{\zeta_-^0}^{\zeta} (\mu \frac{\zeta}{1+\zeta}-\mu \frac{\zeta_-^0}{1+\zeta_-^0})(\zeta/\zeta_-^0)^{2\mu/\sigma^2-2}d\zeta.
\]
The terminal conditions are met by $W^0$, because $G$ is continuous on $(-\infty, -\frac{1}{1-\varepsilon})$. Also, for each $k$, $k=1,2,\dots$, by assumption the HJB equation \eqref{eq: HJBx} is satisfied. Non-negativity is preserved by taking limits, hence, $(\hat W(\zeta;0), \lambda_0)$ satisfies the HJB equation as well. The verification arguments in the proof of Proposition \ref{prop3} imply that the trading strategies associated with the intervals $[\pi_-(\gamma),\pi_+(\gamma)]$ are not only optimal for risk-aversion levels $\gamma\in [0,\bar \gamma]$, but also $[\pi_-^0,\pi_+^0]$ is optimal for a risk-neutral investor.

$\zeta_-(\gamma)$ can have only one accumulation point for $\gamma\downarrow 0$, because $\lambda_0=h(\zeta_-^0)$ is the value function. Uniqueness of $\zeta_-^0$ is therefore clear  and it follows that $\zeta_-^0=\zeta_-(0)$. By assumption, the free boundary problem has a unique solution, hence it follows that $\pi_+(0)=\pi_+^0$. In particular,
the curves  $(0,\bar\gamma]\rightarrow \mathbb R:\gamma\mapsto\pi_{\pm}(\gamma)$  each have a unique limit $\pi_\pm^0$ as $\gamma\downarrow 0$, which
equals $\pi_\pm(0)$, the solution of the free boundary problem.
\end{proof}

\vspace{-1cm}
\bibliographystyle{agsm}

\end{document}